%% file: main.tex
\documentclass{tlp}
\usepackage[utf8]{inputenc}

\usepackage{booktabs}
\usepackage[T1]{fontenc}
\usepackage{newtxtext,newtxmath}
\usepackage[scaled=0.80]{beramono}
\usepackage{microtype}
\usepackage{csquotes}

\usepackage{upquote}
\let\strong\textbf
\usepackage[inline]{enumitem}
\usepackage{amsmath} 
\usepackage[justification=centering]{caption}
\usepackage{subcaption}
\newsavebox{\mybox} 

\input{symbols}

\newlist{andlist}{enumerate*}{1}
\setlist[andlist,1]{
  label={\textit{\roman*)}},
  itemjoin={{, }},
  itemjoin*={{, and }},
}
\newlist{orlist}{enumerate*}{1}
\setlist[orlist,1]{
  label={\textit{\roman*)}},
  itemjoin={{, }},
  itemjoin*={{, or }},
}

\usepackage{newfloat}
\usepackage[
  font=footnotesize,
  labelfont=bf,
  singlelinecheck=false,
]{caption}
\usepackage{wrapfig}

\usepackage[sort]{cite}
\renewcommand\cite[1]{\citep{#1}}
\usepackage[hidelinks]{hyperref}
\usepackage[capitalize,nameinlink]{cleveref}

\crefname{evalQ}{Query}{Query} 

\usepackage[epsilon]{backnaur} 
\let\oldbnfts\bnfts
\renewcommand\bnfts[1]{\oldbnfts{\textbf{#1}}}
\usepackage{listings}
\newcommand\codefont{\ttfamily\fontsize{7.5}{8.5}\selectfont}
\lstdefinelanguage{SPARQL}{
  keywords={SELECT, WHERE, OPTIONAL, FOLLOW, INCLUDE, WITH, SUBWEBS, SERVICE, FILTER, RECURSE, UNION, NOT, EXISTS},
  comment=[f]{\#},
}
\lstdefinelanguage{Turtle}{
  comment=[f]{\#},
}
\crefname{lstlisting}{Listing}{Listings}

\DeclareFloatingEnvironment[name=Document,within=none]{rdfdoc}
\crefname{rdfdoc}{Document}{Documents}
\renewenvironment{rdfdoc}[2]{
  \newcommand\rdfdoclabel{#1}
  \newcommand\rdfdocurl{#2}
}{
  \vspace{-1.25\baselineskip}
  \captionof{rdfdoc}{Contents of \url{\rdfdocurl}}
  \label{\rdfdoclabel}
}

\DeclareFloatingEnvironment[name=Query,within=none]{query}
\crefname{query}{Query}{Queries}
\renewenvironment{query}[2]{
  \newcommand\querylabel{#1}
  \newcommand\querycaption{#2}
}{
  \vspace{-1.25\baselineskip}
  \captionof{query}{\querycaption}
  \label{\querylabel}
}

\DeclareFloatingEnvironment[name=Eval. Query,within=none,placement=ht]{test}
\crefname{test}{Eval. Query}{Queries}
\renewenvironment{test}[2]{
\newcommand\testlabel{#1}
\newcommand\testcaption{#2}
}{
\vspace{-1.0\baselineskip}
\captionof{test}{\testcaption}
\label{\testlabel}
}

\DeclareFloatingEnvironment[name=Results,within=none]{results}
\crefname{results}{Results}{Results}
\renewenvironment{results}[2]{
  \newcommand\resultslabel{#1}
  \newcommand\resultscaption{#2}
}{
  \vspace{-1.5\baselineskip}
  \captionof{results}{\resultscaption}
  \label{\resultslabel}
}

\usepackage{array}
\let\currentfirstcolstyle\relax
\newcommand\firstcolstyle[1]{%
  \gdef\currentfirstcolstyle{#1}%
  #1\ignorespaces%
}
\newcommand\rowstyle[1]{%
  \gdef\currentrowstyle{#1}%
  #1\ignorespaces%
}
\newcolumntype{+}{>{%
  \global\let\currentrowstyle\relax%
  \currentfirstcolstyle%
}}
\newcolumntype{^}{>{%
  \currentrowstyle%
}}

\usepackage{xspace}
\newcommand\Acronym[1]{%
  \expandafter\def\csname#1\endcsname{{\scshape #1}\xspace}%
  \expandafter\def\csname#1s\endcsname{{\scshape #1}s\xspace}%
}
\newcommand\AcronymNum[3]{%
  \expandafter\def\csname#1\endcsname{{\scshape #2$_{#3}$}\xspace}%
  \expandafter\def\csname#1s\endcsname{{\scshape #2$_{#3}$}s\xspace}%
}
\Acronym{api}
\Acronym{http}
\Acronym{ldf}
\Acronym{ldp}
\Acronym{ldql}
\Acronym{ltqp}
\Acronym{rdf}
\Acronym{shacl}
\Acronym{sparql}
\Acronym{scl}
\Acronym{swsl}
\AcronymNum{swslOne}{swsl}{1}
\AcronymNum{swslTwo}{swsl}{2}
\newcommand\dbpedia{\textsc{db}pedia\xspace}

\newcommand\URI {{\scshape uri}\xspace}
\newcommand\URIs{{\scshape uri}s\xspace}
\newcommand\URL {{\scshape url}\xspace}

\newcommand\WOLD{{\scshape wold}\xspace}
\newcommand\caWOLD{sa-{\scshape wold}\xspace}
\newcommand\WOLDs{{\scshape wold}s\xspace}

\newcommand\IRI {{\scshape iri}\xspace}
\newcommand\IRIs{{\scshape iri}s\xspace}

\include{reviewing}

\usepackage{tikz}
\usetikzlibrary{positioning}
\tikzset{main node/.style={circle,fill=black!10,draw,minimum size=1cm,inner sep=0pt},
            }
\usetikzlibrary{arrows}
\usetikzlibrary{shapes}
\usepackage{float}
\usepackage{multirow,tabularx}

\newtheorem{theorem}{Theorem}

\newtheorem{definition}{Definition}
\newtheorem{remark}{Remark}

\renewcommand\cite[1]{\citep{#1}}

\begin{document}

\lefttitle{Distributed Subweb Specifications for Traversing the~Web}

\jnlPage{1}{8}
\jnlDoiYr{2021}
\doival{10.1017/xxxxx}

\title[Distributed Subweb Specifications for Traversing the Web]{Distributed Subweb Specifications for Traversing the~Web
\thanks{
This research received funding from the Flemish Government under the ``Onderzoeksprogramma Artifici\"ele Intelligentie (AI) Vlaanderen'' programme and SolidLab Vlaanderen (Flemish Government, EWI and RRF project VV023/10).
Ruben Taelman is a postdoctoral fellow of the Research Foundation -- Flanders (FWO) (1274521N). Heba Aamer is supported by the Special Research Fund (BOF) (BOF19OWB16)
}
}

\begin{authgrp}
\author{\sn{Bart} \gn{Bogaerts}}
\affiliation{Vrije Universiteit Brussel, Belgium}
\author{\sn{Bas} \gn{Ketsman}}
\affiliation{Vrije Universiteit Brussel, Belgium}
\author{\sn{Younes} \gn{Zeboudj}}
\affiliation{Vrije Universiteit Brussel, Belgium}
\author{\sn{Heba} \gn{Aamer}}
\affiliation{Universiteit Hasselt, Hasselt, Belgium}
\author{\sn{Ruben} \gn{Taelman}}
\affiliation{Ghent University -- imec -- IDLab}
\author{\sn{Ruben} \gn{Verborgh}}
\affiliation{Ghent University -- imec -- IDLab}
\end{authgrp}

\history{\sub{xx xx xxxx;} \rev{xx xx xxxx;} \acc{xx xx xxxx}}

\maketitle


\begin{abstract}
{
Link Traversal--based Query Processing (\ltqp),
in which a~\sparql~query
is evaluated over a~web of documents
rather than a~single dataset,
is often seen as a~theoretically interesting yet impractical technique.
However,
in~a~time where the hypercentralization of data
has increasingly come under scrutiny,
a~decentralized Web of Data
with a~simple document-based interface is appealing, as
it enables data publishers to control their data and access rights.
While \ltqp allows evaluating complex queries over such webs,
it suffers from performance issues
(due to the high number of documents containing data)
as well as information quality concerns
(due to the many sources providing such documents).
In existing \ltqp approaches, the burden of finding  sources to query is
entirely in the hands of the \emph{data consumer}.
In this paper, we argue that to solve these issues, \emph{data publishers}
should also be able to suggest sources of interest and \emph{guide} the data consumer
towards relevant and trustworthy data.
%
%
%
%
We introduce a theoretical framework that enables such guided
link traversal
and study its properties.
We illustrate with a theoretic example that this can improve query results and  reduce
the number of network requests.
We evaluate our proposal experimentally on a virtual linked web with specifications and indeed observe that not just the data quality but also the efficiency of querying improves.
}
\end{abstract}

\input{introduction}
\input{use-case}
\input{requirements}
\input{relwork}
\input{formalization-new}
\input{syntax-new}

\input{ldql-requirements}
\input{experiments}
\input{query-processing}
\input{conclusion}

\subsection*{Acknowledgements}
We thank the reviewers for their thorough review and comments
on the earlier version of this paper.


\bibliographystyle{tlplike}
\bibliography{references}

%
\appendix

\section{Virtual linked web annotation}
\input{appendix.subwebs}

\end{document}

%% file: symbols.tex
\newcommand\cAll{\ensuremath{c_\textsf{All}}\xspace}
\newcommand\cNone{\ensuremath{c_\textsf{None}}\xspace}
\newcommand\cMatch{\ensuremath{c_\textsf{Match}}\xspace}

\newcommand\m[1]{\ensuremath{#1}\xspace}
\newcommand\mcl[1]{\ensuremath{\mathcal{#1}}\xspace}

\newcommand\T{\ensuremath{\mathcal{T}}\xspace}
\newcommand\tripleset{\m{2^\T}}
\newcommand\alldocs{\ensuremath{\mathcal{D}}\xspace}
\newcommand\doc{\ensuremath{d}\xspace}
\newcommand\docs{\ensuremath{D}\xspace}
\newcommand\allwolds{\ensuremath{\mathcal{W}}\xspace}
\newcommand\wold{\ensuremath{W}\xspace}
\newcommand\data{\m{\mathit{data}}\xspace}

\newcommand\adoc{\m{\mathit{adoc}}}
\newcommand\uri{\m{u}}
\newcommand\uris{\m{\mathcal{U}}}

\newcommand\blanks {\m{\mathcal{B}}}
\newcommand\literals{\m{\mathcal{L}}}

\newcommand\ignore[1]{}

\newcommand\simpl{\m{\mathit{singleton}}} 
\newcommand\ltrue{\m{\mathrm{true}}}
\newcommand\lfalse{\m{\mathrm{false}}}
\newcommand\alltriples{\mcl{T}}

\newcommand\tuple[1]{\m{\langle #1 \rangle}}
\newcommand\triple[3]{\m{(#1, #2, #3)}} 
\newcommand\urisof{\m{\mathrm{uris}}}
\newcommand\evaluation[2]{\m{[\![#1]\!]^{#2}}}

\newcommand\ctxqry[1][{\cw}]{\m{\mathbb{P}}}
\newcommand\sparqlex{\m{P}}

\newcommand\subwebs[1]{\m{\mathit{subwebs}(#1)}}

\newcommand\cw{\m{\Theta}}

\newcommand\cwtuple{\m{\boldsymbol{\cw}}}
\newcommand\ctxwold[1][]{\m{\mathbb{W}_{#1}}}
\newcommand\sselect{\sigma}

\newcommand\queryingagent{querying~agent\xspace}
\newcommand\queryprocessor{query~processor\xspace}
\newcommand\datapublisher{data~publisher\xspace}
\newcommand\queryingagents{\queryingagent{s}\xspace}

\newcommand\datapublishers{\datapublisher{s}\xspace}

\newcommand\ldql{\m{\mathit{LDQL}}}
\newcommand\lpe{\textsc{lpe}\xspace}
\newcommand\lpes{\textsc{lpe}s\xspace}

\DeclareMathOperator \ldqlAND{AND}
\newcommand\seed{\m{s}}

\newcommand{\brak}[1]{\m{\{#1\}}}
\newcommand\lpeEval[3]{\m{\llbracket#1\rrbracket_{{#2}}^{{\m{#3}}}}}
\newcommand\ldqlEval[3]{\m{\llbracket#1\rrbracket_{{#2}}^{{\m{#3}}}}}
\newcommand\queryEval[3]{\m{\llbracket#1\rrbracket_{{#2}}^{{\m{#3}}}}}
\newcommand\constrEval[2]{\m{\llbracket#1\rrbracket^{#2}}}
\newcommand\GGPEval[2]{\m{\llbracket#1\rrbracket^{#2}}}

\DeclareMathOperator{\lpeOr}{\text{\textbf{|}}}

\newcommand{\uctx}{\m{u_{\mathit{ctx}}}}

\newcommand{\lstinlineProxy}[1]{\lstinline{#1}}

\newcommand{\soi}{\m{\mathit{soi}}}

%% file: reviewing.tex
\usepackage[dvipsnames,svgnames]{xcolor} 
\usepackage[normalem]{ulem} 

\newcommand\rt[1]{{\color{JungleGreen}\textbf{RT:} #1}}

\newcommand\bartm[1]{  \marginpar{\color{Cerulean}\textbf{BB:} #1}}

\makeatletter
\font\uwavefont=lasyb10 scaled 700
\def\spelling{\bgroup\markoverwith{\lower3.5\p@\hbox{\uwavefont\textcolor{Red}{\char58}}}\ULon}
\def\formalism{\bgroup\markoverwith{\lower3.5\p@\hbox{\uwavefont\textcolor{Orange}{\char58}}}\ULon}
\def\grammar{\bgroup\markoverwith{\lower3.5\p@\hbox{\uwavefont\textcolor{LimeGreen}{\char58}}}\ULon}
\def\phrasing{\bgroup\markoverwith{\lower3.5\p@\hbox{\uwavefont\textcolor{RoyalBlue}{\char58}}}\ULon}

\newcommand\remove{\bgroup\markoverwith{\textcolor{red}{\rule[0.5ex]{2pt}{0.4pt}}}\ULon}
\makeatother

%% file: introduction.tex
\section{Introduction}
\label{sec:Introduction}


The World-Wide Web provides a~permissionless information space
organized as interlinked documents.
The Semantic Web builds on~top~of~it
by~representing data
in a~machine-interpretable format,
fueled~by the Linked Data principles.
In contrast to more complex data-driven \apis,
the simplicity of document-based interfaces comes with multiple advantages.
They scale easily, and can be hosted on many different kinds of hardware and software;
we can realize the \emph{\enquote{anyone can say anything about anything}} principle
because every publisher has their own domain in the Web,
within which they can freely refer to concepts from other~domains;
and complex features such as access control or versioning
are technically easy to achieve on a~per-document basis.

However,
decentralized interfaces are notoriously more difficult to query.
As such,
the past decade has instead been characterized
by Big Data and hypercentralization,
in which data from multiple sources
becomes aggregated in an increasingly smaller number of sources.
While extremely powerful from a~query and analytics perspective,
such aggregation levels lead to a~loss of control and freedom
for individuals and small- to medium-scale data~providers.
This in~turn has provoked some fundamental
legal, societal, and economical questions
regarding the acceptability 
of such hypercentral platforms.
As~such, there is again an increasing demand for more decentralized systems,
where data is stored closer to its authentic~source,
in line with the original intentions of the Web~\cite{verborgh_timbl_chapter_2020}.

As with Big Data,
query processing on the Semantic Web has traditionally
focused on \emph{single} databases.
The \sparql query language
allows querying such a single \rdf~store through the \sparql protocol,
which places significantly more constraints on the server
than a~document-based interface~\cite{verborgh_jws_2016}.
While \emph{federated} query processing enables
incorporating data from multiple \sparql~endpoints,
federated queries have very limited link traversal capabilities and \sparql endpoints easily
experience performance degradation~\cite{SparqlReadyForAction}.

Fortunately,
a~technique was introduced to query webs of data:
\emph{Link Traversal--based Query Processing} (\ltqp)~\cite{Hartig2009,Hartig2013},
in which an agent evaluates a~\sparql query
over a~set of documents that is continuously expanded
by \emph{selectively} following hyperlinks inside of~them.
While \ltqp demonstrates the independence of queries and selection of sources (on which these queries
need to be executed),
it has mostly remained a~theoretical exercise,
as its slow performance
makes it unsuitable for practical purposes.
The fact that \ltqp can yield more results
than single-source query evaluation,
gave~rise to~different notions
of \emph{query semantics} and \emph{completeness}~\cite{Hartig2012}.
While more data can be considered advantageous, it can also lead to doubts regarding
\emph{data~quality},
\emph{trustworthiness},
\emph{license compatibility},
or \emph{security}~\cite{ldtraversalsecurity}.
Together with performance,
these concerns seem to have pushed \ltqp to the~background.

In this article, we identify two limitations of existing \ltqp approaches.
Essentially, all existing \ltqp approaches identify a \emph{subweb} of the web of linked data on which a query needs to be executed. The first limitation is that \emph{the responsibility for defining how to construct this subweb
is entirely in the hands of the data consumer, from now on referred to as the \strong{\queryingagent}} (which can be an end-user or machine client). In other words, existing approaches make the assumption that the \queryingagent can determine
perfectly which links should be traversed. However, since every data publisher
can freely choose how to organize their data, we cannot expect a single agent to possess complete knowledge of how such traversals should proceed.
A second restriction is that current \ltqp formalisms provide an all-or-nothing approach: a document is either included in the subweb of interest in its entirety, or not at all, while for data-quality reasons, it would be useful to only take parts of documents into account.
For instance, an academic who has moved institutions might specify that the data provided by institution A is trustworthy up to a certain date and that for later information about them, institution B should be consulted. More radically, a certain end user might wish to specify that Facebook's data about who her friends are is correct, without thereby implying that any triple published by Facebook should be taken into account when performing a query.

In this paper, building on the use case of the next section, we propose an approach for \emph{guided} link traversal that overcomes these two limitations.
In our proposal, each data publisher has their own subweb of interest, and publishes a specification of how it can be constructed. They can use this for instance to describe the organization of their data,
or to describe parties they trust (as well as for which data they trust them).
The data consumer can then construct a subweb of interest \emph{building on} the subwebs of the publishers, e.g., deciding to include parts of a subweb, or to omit it. 
As such, the data publishers \emph{guide} the data consumer towards relevant data sources.
%
%
We~focus on the theoretical foundations
and highlight opportunities for result quality and performance improvements.
We implemented our proposal and experimentally validated it on a crafted web of linked data, annotated with subweb specifications in our formalism.

The rest of this paper is structured as follows.
In \cref{sec:usecase}, we present a motivating use case; afterwards, in \cref{sec:prelims} we recall some basic definitions.
From our use case, in \cref{sec:requirements} we extract several desired properties.
Related work is discussed in light of the use case and the derived desired properties in \cref{sec:related}.
Our theoretical formalism is presented in \cref{sec:formalism}, and a concrete web-friendly syntax for it is discussed in \cref{sec:syntax}.
In \cref{sec:ldql}, we investigate to which extent existing link traversal formalisms can ``simulate'' the behaviour of our formalism. The answer is that even for very simple expressions, such simulations are not possible, thereby illustrating the expressive power of our new formalism.
We evaluate the effect of using subweb specifications on performance and on query result quality in \cref{sec:experiments}.
We end the paper with a discussion and a conclusion.

\paragraph{Publication History}
A short version of this paper was presented at the 2021 RuleML+RR conference \cite{ruleml/BogaertsKZATV21}.
This paper extends the short version with proofs and an experimental evaluation.

%% file: use-case.tex
\section{Use Case}
\label{sec:UseCase}\label{sec:usecase}
As a~guiding example throughout this article,
we~introduce example data and queries for a~use~case that
stems from the Solid ecosystem~\cite{verborgh_timbl_chapter_2020},
where every person has their own \emph{personal data vault}.
%
\begin{figure}[tb]
%
%
\begin{minipage}[t]{.49\linewidth}
\begin{rdfdoc}{lst:Uma}{https://uma.ex/}
\begin{lstlisting}
<https://uma.ex/#me> foaf:knows
  <https://ann.ex/#me>, <https://bob.ex/#me>.
<https://bob.ex/#me> foaf:img <bob.jpg>.
\end{lstlisting}
\end{rdfdoc}
\begin{rdfdoc}{lst:Ann}{https://ann.ex/}
\begin{lstlisting}
<https://ann.ex/#me> foaf:isPrimaryTopicOf <https://corp.ex/ann/>.
<https://ann.ex/#me> foaf:weblog <https://ann.ex/blog/>.
<https://ann.ex/#me> foaf:maker <https://photos.ex/ann/>.
\end{lstlisting}
\end{rdfdoc}
\begin{rdfdoc}{lst:Bob}{https://bob.ex/}
\begin{lstlisting}
<https://bob.ex/#me> foaf:name "Bob";
  foaf:mbox <mailto:me@bob.ex>;
  foaf:img <funny-fish.jpg>.
<https://uma.ex/#me> foaf:knows
  <http://dbpedia.org/resource/Mickey_Mouse>.
<https://ann.ex/#me> foaf:name "Felix".
\end{lstlisting}
\end{rdfdoc}
\end{minipage}
%
%
\begin{minipage}[t]{.49\linewidth}
\begin{rdfdoc}{lst:AnnDetails}{https://corp.ex/ann/}
\begin{lstlisting}
<https://ann.ex/#me> foaf:name "Ann";
  foaf:mbox <mailto:ann@corp.ex>;
  foaf:img <me.jpg>.
\end{lstlisting}
\end{rdfdoc}
\vspace{5em}
\vspace{0.3pt}
\begin{query}{qry:Friends}{Application query in \sparql}
\begin{lstlisting}[language=SPARQL]
SELECT ?friend ?name ?email ?picture WHERE {
  <https://uma.ex/#me> foaf:knows ?friend.
  ?friend foaf:name ?name.
  OPTIONAL { ?friend foaf:mbox ?email.
             ?friend foaf:img  ?picture. }
}
\end{lstlisting}
\end{query}
\end{minipage}
%
%
%
\medskip
\begin{results}{tbl:Results}
  {%
    Possible results of \ltqp of the query in \cref{qry:Friends}
    with \url{https://uma.ex/} as seed
  }
  \codefont
  \begin{tabular}{|+r|^l|^l|^l|^l|}
    \toprule
    \rowstyle{\bfseries}
    \firstcolstyle{\bf}
      & ?friend               & ?name             & ?email               & ?picture                        \\
    \midrule
    1 & <https://ann.ex/\#me> & "Ann"             & <mailto:ann@corp.ex> & <https://corp.ex/ann/me.jpg>    \\
    2 & <https://bob.ex/\#me> & "Bob"             & <mailto:me@bob.ex>   & <https://uma.ex/bob.jpg>        \\
    3 & <https://bob.ex/\#me> & "Bob"             & <mailto:me@bob.ex>   & <https://bob.ex/funny-fish.jpg> \\
    4 & <https://ann.ex/\#me> & "Felix"           & <mailto:ann@corp.ex>  & <https://corp.ex/ann/me.jpg>    \\
    5 & dbr:Mickey\_Mouse     & "Mickey Mouse"@en & \itshape NULL        & \itshape NULL                   \\
    \bottomrule
  \end{tabular}
\end{results}
\end{figure}
Let us consider 3~people's profile documents,
stored in their respective data vaults.
Uma's profile (\cref{lst:Uma}) lists her two friends Ann and Bob.
Ann's profile (\cref{lst:Ann}) contains links
to her corporate page and various other pages.
Bob, a~self-professed jokester,
lists his real name and email address in his profile (\cref{lst:Bob}),
in addition to a~funny profile picture
and a~couple of factually incorrect statements
(which he is able to publish given the open nature of the~Web).
Note how Ann provides additional facts about herself
into the external document she links to (\cref{lst:AnnDetails}),
and Uma's profile suggests a~better profile picture for Bob (\cref{lst:Uma}).

Next,
we consider an \emph{address book} application
that displays the details of a~user's contacts.
At design-time,
this application is unaware of the context and data distribution
of the user and their friends.
If we assume Uma to be the user,
then the application's data~need can be expressed as \cref{qry:Friends},
which is a~generic \sparql template
in which only the \URL corresponding to Uma's identity (\url{https://uma.ex/#me})
has been filled~out.

With traditional \ltqp
(under \cAll~semantics~\cite{Hartig2012}),
results include those in \cref{tbl:Results}.
However,
the actually desired results are Rows 1 and~2,
which contain Uma's two friends
with relevant details.
Rows~3--5 are formed using triples that
occur in Bob's profile document
but are not considered trustworthy by Uma
(even though other triples in the same document are).
To obtain these results,
a~query engine would need to fetch at least 7~documents:
the profile documents of the 3~people (Uma, Ann, Bob),
the 3~documents referred to by Ann's profile (\cref{lst:Ann}),
and the \dbpedia~page for Mickey~Mouse.


%% file: requirements.tex

\section{Preliminaries}\label{sec:prelims}
As a~basis for our data model of a~Web of Linked Data,
we use the \rdf data model~\cite{RDF}.
That is,
we assume three pairwise disjoint, infinite sets:
\uris (for \URIs),
\blanks (for blank~nodes),
\literals(for literals).
An \emph{\rdf triple} is a~tuple $\triple{s}{p}{o} \in \T$,
with \T the set of all~triples defined as
\[\T = (\uris\cup\blanks)\times\uris\times(\uris\cup\blanks\cup\literals);\] if $t=\triple{s}{p}{o}\in\T$, then $\urisof(t) = \{s,p,o\}\cap\uris$.
A set of triples is called a \emph{triple graph} or an \emph{\rdf graph}. An \emph{\rdf dataset} is a set of tuples \brak{\tuple{n_i,g_i}} such that $n_i\in \uris$ and $g_i$ an \rdf graph, where $g_0$ is referred to as the \emph{default graph}.

We assume another set \alldocs,
disjoint from the aforementioned sets \uris, \blanks and \literals,
whose elements are referred to as \emph{documents}.
The \rdf~graph contained in each document
is modeled by a~function $\data: \alldocs \rightarrow \tripleset$
that maps each document to a~finite set of triples.

\begin{definition}
	\label{defW}
  A \emph{Web of Linked Data (\WOLD)} $\wold$
  is a~tuple \tuple{\docs, \data, \adoc}
  where $\docs$ is a~set of documents $\docs\subseteq \alldocs$,
  \data a function from $\docs$ to $\tripleset$
  such that $\data(\doc)$ is finite for each $\doc \in \docs$,
  and $\adoc$ a~partial function from $\uris$ to $\docs$.
  If \wold is a \WOLD,
  we use $\docs_\wold$, $\data_\wold$, and $\adoc_\wold$ for its respective components.
  The set of all \WOLDs is denoted \allwolds.
\end{definition}

We aim to define parts of a~web as~subwebs.
While existing definitions only consider
the inclusion of documents in their entirety~\cite{Hartig2012},
we allow for \emph{partial} documents
to enable fine-grained control about which data is to be used for answering certain queries.

\begin{definition}
	Consider two \WOLDs
  $\wold = \tuple{\docs, \data, \adoc}$
	and
  $\wold'=\tuple{\docs', \data', \adoc'}$.
  We~say that $\wold'$ is a \emph{subweb} of \wold if
	\begin{enumerate}
		\item $\docs'\subseteq \docs$
		\item $\forall \doc \in \docs' : \data'(\doc) \subseteq \data(\doc)$
		\item $\adoc'(u) =\adoc(u)$ if $\adoc(u)\in\docs'$
      and $\adoc'(u)$ is undefined otherwise.
	\end{enumerate}
  We write \subwebs{\wold} for the set of subwebs~of~\wold.
\end{definition}

The simplest type of subwebs are those only consisting of a single document.
\begin{definition}
    Let \wold be a \WOLD and $d\in\docs$. We use $\simpl(d,\wold)$ to denote the (unique) subweb $\tuple{\{d\}, \data', \adoc'}$ of \wold with $\data'(d)=\data(d)$.
\end{definition}
Additionally, if two subwebs of a given \WOLD are given, we can naturally define operators such as union and intersection on them; in this paper, we will only need the union.
\begin{definition}
 If $\wold_1$ and $\wold_2$ are subwebs of $\wold$, we define $\wold_1\cup\wold_2$ to be the unique subweb  $\tuple{\docs', \data', \adoc'}$ of \wold with
 \begin{itemize}
  \item $\docs' = \docs_{\wold_1}\cup\docs_{\wold_2}$, and
  \item $\data'(d) = \data_{\wold_1}(d) \cup \data_{\wold_2}(d)$  for each $d\in\docs'$, where, slightly abusing notation,  we use $\data_{\wold_i}(d)=\emptyset$ if $d\not\in\docs_{\wold_i}$.
 \end{itemize}
\end{definition}

%
%
%

\section{Requirements}
\label{sec:requirements}

From the use case, we extracted four requirements that motivate our definitions.  
 \paragraph{A Declarative Language for Selecting Data Sources}
Similar to existing \ltqp approaches, we need  a language to describe which data sources to select (possibly starting from a given seed).
 We want such a   language to be declarative, i.e., focus on \emph{which} sources to use, rather than \emph{how}
 to obtain them.
 Formally, we expect a source selection expression to evaluate in a given WOLD to a set of \URIs representing the documents to be included.
\paragraph{Independence of Query and Subweb Specification}
Motivated by principles of reusability and separation of concerns, we want
the \emph{query} to be formulated independently from the \emph{subweb over which the query is to be evaluated}.
While it might --- to a certain extent --- be possible to encode traversal directions in (federated) \sparql queries, \emph{what do I want to know} and \emph{where do I want to get my information} are two orthogonal concerns that we believe should be clearly separated, in order to improve readability,  maintainability, and reusability.
For example, in the use case, the phone book application defines the \emph{query}, while Uma defines her own \emph{subweb of interest} (consisting of her own document, as well as parts of the documents of her friends).
The application should be able to run with different subwebs (e.g., coming from other users), and Uma's subweb of interest should be reusable in other applications.
\paragraph{Scope Restriction of Sources}
One phenomenon that showed up in the use case is that we want to trust a certain source, but only for specific data.
We might for instance want to use all our friends' data sources, but only to provide information about themselves.
This would avoid ``faulty'' data providers such as Bob from publishing data that pollute up the entire application, and it would give a finer level of control over which data is to be used to answer queries.
On the formal level, this requirement already manifests itself in the definition of \emph{subweb} we chose: contrary to existing definitions \cite{Hartig2012}, we allowed a document in a subweb to have only a subset of the data of the original document.
\paragraph{Distributed Subweb Specifications}
Finally, we arrive at the notion of distribution.
This is the feature in which our approach most strongly deviates from the state-of-the-art in link traversal.
While the semantic web heavily builds on the assumption that \emph{data} is decentralized and different agents have different pieces of data to contribute, existing link traversal--based approaches still assume that the \emph{knowledge of where this data can be found} is completely in the hands of the querying agent at query time, or at least that the \emph{principles by which the web has to be traversed} can be described by the querying agent.
However, as our use case illustrates, this is not always the case: Ann decided to distribute her information over different pages; the agent developing the phone book application cannot possibly know that the triple
\lstinline{<https://ann.ex/#me> foaf:isPrimaryTopicOf <https://corp.ex/ann/>.} indicates that information from \lstinline{<https://corp.ex/ann/>} is ``equally good'' as information from Ann's main document.
Stated differently, only Ann knows how her own information is organized and hence if we want to get personal information from Ann, we would want her to be able to describe herself how or where to find this data.
To summarize, we aim to allow \emph{document publishers to publish specifications of subwebs in the same declarative language as used by the querying agents} and to allow \emph{querying agents to decide whether or not to include the data from such subwebs}.

Moreover, the fact that published subweb specifications can be used across different applications forms an incentive for document publishers to actually build and publish such specifications. 
For instance in our running example, if Uma publishes her subweb specficiations for the purpose of the described phone book application, and later she wants to use a different application, e.g., a social media app, to access her friends (and their data), she can use the same specifications and simply link to her own data pod.


\ignore{
We further assume, for the sake of simplicity,
that every document can be traced back to a~\datapublisher,
who makes a~description of their context available.
This description could state,
in terms of subwebs,
for what kind of data
the \datapublisher considers certain sources or documents
authoritative, trustworthy, license-compatible, etc.
For~example,
Ann could state that she considers any document under the domain \url{corp.ex}
to be authoritative for triples of the form
$\triple{\textit{https://ann.ex/\#me}}{p}{o}$,
or in~fact any other graph pattern or constraint of her choice.
Therefore,
it makes sense to define how a~context
yields subwebs per document,
containing the data that it considers in~scope.
}

%% file: relwork.tex
\section{Related Work}
\label{sec:RelWork}
\label{sec:related}

\ignore{
\paragraph{Linked Data on the Web}\bartm{Can we move this section to the back of the paper? No need to have it now, it seems}\bartm{Is this paragraph relevant for the paper? THe intro contains enough motivation for linked data already... }\rt{It gives a nice intro to the next LTQP paragraph, but I'm ok with removing it if we need the space. Not sure if we want to move it to the back of the paper though. I feel like this section makes our motivations a bit stronger based on how it is written now. But I have no strong opinion on this.}
Publishing \rdf data following the Linked Data principles
results in each resource being described in a~separate document,
which links to resources in other documents
that can be retrieved via \http requests.
Since this approach embraces the fundamental Web principles,
it can be hosted with only limited~effort.
As~a~result,
this document-oriented interface can realize orthogonal features,
such as granular authentication or versioning,
on the \http~level.
While Linked Data is often associated with open~data,
cases for personal, private, or
otherwise restricted environments exist~\cite{verborgh_jws_2016}.
These restrictions are interesting
because merging such sources before querying
might be disallowed for~legal or other reasons;
hence, those sources need to be queried in-place.
}

\paragraph{Web Decentralization}
To counter the various issues surrounding centralized data management,
efforts such as Solid~\cite{verborgh_timbl_chapter_2020}, Mastodon~\cite{mastodon}, and others~\cite{decentralizednanopubs}
aim to decentralize data on the Web.
While approaches such as Mastodon aim to split up data into several federated instances,
Solid takes a more radical approach, where each person can store data in a personal data vault,
leading to a wider domain of decentralization.
Solid is built on top of a collection of open Web standards~\cite{spec:solidprotocol},
which makes it an ecosystem in which decentralized applications can be built.
This includes specifications on how data can be exposed through the HTTP protocol~\cite{spec:ldp},
how to manage authentication~\cite{spec:solidoidc} and authorization~\cite{spec:wac,spec:acp},
and representing identity~\cite{spec:webidprofile}.
Our approach for enabling the publication of subwebs of interest has precedent in the Solid ecosystem,
since approaches such as Type Indexes~\cite{spec:typeindex} and Shape Trees~\cite{spec:shapetrees} already exist
that enable data discovery in data vaults by type or data shape.
Our approach differs from these approaches in the fact that we use this published information
for link \emph{pruning} instead of link \emph{discovery}.

\paragraph{Link Traversal-based Query Processing}
Over a~decade ago, the paradigm of Link Traversal-based Query Processing was introduced~\cite{Hartig2009},
enabling queries over document-oriented interfaces.
The main advantage of this approach is that queries can always be executed over live data,
as opposed to querying over indexed data that may be stale.
The main disadvantages of this approach are that query termination and result completeness are not guaranteed,
and that query execution is typically significantly slower
than database-centric approaches such as \sparql endpoints.
Several improvements have been suggested to cope with these problems~\cite{Hartig2013}.
For example, the processing order of documents can be changed so that certain documents are \emph{prioritized}~\cite{WalkingWithoutMap},
which allows relevant results to be emitted earlier in an iterative manner~\cite{Squin},
but does not reduce total execution time.
In this work,
we propose to tackle this problem by allowing publishers to specify their subweb of interest.
These specifications are then used to
\emph{guide} the query engine towards relevant (according to the data publishers at hand) documents.
%
\ltqp is related to the domain of focused crawlers~\cite{focusedcrawling,focusedcrawlingimproving},
which populate a local database by searching for specific topics on Web pages.
It is also related to SQL querying on the Web~\cite{infogatheringwwww3ql,queryingwwwsql},
which involves querying by attributes or content within Web pages.
In contrast to these two related domains, \ltqp is based on the RDF data model,
which simplifies data integration due to universal semantics. 

\paragraph{Reachability Semantics}
The \sparql query language
was originally introduced for query processing over \rdf~databases.
Since \ltqp involves a~substantially different kind of sources,
a family of new semantics was introduced~\cite{Hartig2012},
involving the concept of a~\emph{reachable}~subweb.
When executing a query over a set of \emph{seed~documents},
the reachable~Web is the set of documents
that can be reached from these seeds
using one of different \emph{reachability~criteria}.
These criteria are functions
that test each data triple within retrieved documents,
indicating which (if any) of the \URIs in the triple components
should be dereferenced
by interpreting them as the \URI of a~document
that is subsequently retrieved over~\http.
The simplest reachability criterion is \cNone,
where none of the \URIs from the seed documents is followed,
such that the query will be executed over the union of triples across the seeds.
\Cref{qry:Friends} under \cNone with \cref{{lst:Uma}} as seed
would thus not yield any results.
The~\cAll reachability criterion involves following all encountered \URIs,
which is the strategy in the example of \cref{tbl:Results}.
A more elaborate criterion is \cMatch,
which involves following \URIs from data triples
that match at least one triple pattern from the query.
\cMatch~can significantly reduce the number of traversals compared to \cAll.
However,
evaluating \cref{qry:Friends} with \cMatch semantics
would not yield results for~Ann (rows 1 and~4).
Her details are only reachable
via a~triple with predicate \verb!foaf:isPrimaryTopicOf!,
which does not match any of the query's triple patterns;
hence, the relevant document is never~visited.
So~while \cMatch can lead to better performance,
it comes at the cost of fewer results,
showing that none of these approaches is optimal.

\ignore{
\paragraph{\ltqp Query Syntax}
Since \sparql has not been designed for \ltqp,
it lacks the ability to represent traversal instructions.
To cope with this problem,
several alternative approaches either modify the semantics of \sparql clauses~\cite{SPARQLLD,LDPP},
or introduce non-\sparql-based languages~\cite{ldql,NautiLOD}.
While incorporating traversal into the query language
is advantageous in terms of expressivity,
it requires queries to be overly specific to a~certain context.
In contrast,
decentralized applications typically only know the desired data patterns in advance,
but not the user-specific context.
As~such, we decouple the query and traversal context at~design time,
such that they can be combined by query engines at~runtime.
Moreover, this context can be composed
out of smaller context from distributed actors,
such that no actor needs knowledge about the entire network.
}

\paragraph{Delegation} The concept of subwebs is somewhat related to the presence of active rules in rule-based languages for distributed data management. A particularly relevant project in this context is Webdamlog~\cite{DBLP:conf/pods/AbiteboulBGA11}, a Datalog-based declarative language for managing knowledge on the web with support for rule-delegation. Here, delegation is achieved by allowing rules to get partially materialized by different peers.

%% file: formalization-new.tex
\section{A Formalism for Subweb Specifications}
\label{sec:Formalization}
\label{sec:formal}\label{sec:formalism}
Inspired by the desired properties from \cref{sec:requirements}, we now define a
formalism to describe subwebs of interest.
In our formalism, different agents will be able to provide a description of a \emph{subweb of interest}; they will be able to specify declaratively in (which parts of) which
documents they are interested.
We do not make any assumption here about what the reason for this ``interest'' is; depending on the context at hand, different criteria such as relevance, trustworthiness, or license-compatibility can be used.
Such a description of a subweb of interest can be given by the  \strong{\queryingagent} (an end-user or machine client)
which provides it at runtime to the \strong{\queryprocessor}.
Additionally,
every \strong{\datapublisher} can use the same mechanism
to make assertions about their beliefs,
such that other \datapublishers or \queryingagents can reuse those
instead of requiring explicit knowledge.
For instance,
a~\datapublisher can express which sources they consider
relevant or trustworthy for what kinds of data:
a~researcher might indicate that a~certain source
represents their publication record correctly,
whereas another source captures their affiliation history.
A certain agent
\emph{might} or \emph{might not}
choose to take the subweb of interest of a~\datapublisher into consideration.
In the use case of~\cref{sec:UseCase},
the application generates a query \sparqlex as \cref{qry:Friends},
and end-user Uma expresses that she trusts her own profile for her list of contacts,
and that she trusts those contacts for their own details.
Furthermore, each of these friends can indicate which other documents they trust for which information. For instance, Ann~could express that she trusts \url{corp.ex}
for her personal details.
Essentially, in this case Uma partially \emph{delegates} responsibility of traversing the web to Ann, but
only retains information about Ann from Ann's subweb of interest. 
This leads to the following definitions.

\newcommand{\mydef}{\mathrel{:=}}
\begin{definition}
    A \emph{source selector} is a function $\sselect:\allwolds\to2^\uris$.
\end{definition}
\begin{definition}
    A \emph{filter} is a function  $f: 2^\alltriples \times \uris \to 2^\alltriples$ such that $f(S,u)\subseteq S$ for every  $S  \subseteq \alltriples$ and $u\in\uris$.
    For a \WOLD $\wold=\tuple{\docs, \data, \adoc}$ and \URI $u$;
    we extend the notation and also write $f(\wold, u)$ to denote the subweb
    $ \tuple{\docs, \data', \adoc}$ of $\wold$
    with $\data'(d) \mydef f(\data(d), u)$ for each $d \in \docs$.
\end{definition}
    In our running example, if Uma wants for each of her friends to only include statements they make about themselves, she can use a source selector $\sselect$ that extracts her friends, e.g, with $\sselect(W) = \{ o \mid \triple{s}{ \text{\lstinlineProxy{foaf:knows}}}{o} \in \data(\adoc(s))\text{ with }s = \text{\lstinlineProxy{<https://uma.ex/#me>}}\}$ and with a filter
 that maps $(S,u)$ to $\{\triple{s}{p}{o} \in S \mid s = u\}$.
If we assume that $W$ is a \WOLD in which only a particular friend $u$ of Uma provides triples, then $f(W, u)$ is the subweb of $W$ in which friend $u$ has only the triples making statements about him or herself.


\begin{definition}
 A \emph{subweb specification}, often denoted $\cw$, is a set of tuples of the form $(\sselect,b,f)$, where
  $\sselect$ is a source selector;
$b$ is a Boolean; and
$f$ is a filter.
\end{definition}
Intuitively, the Boolean $b$ in $(\sselect,b,f)$ indicates whether to include for each \URI $u\in \sselect(\wold)$ (the filtered version of) the subweb of $\adoc(u)$ or only $u$'s document.
Finally,
this brings us to the definition of a~specification-annotated~\WOLD (\caWOLD in short): a \WOLD
extended with the knowledge of how to construct the subweb of all \datapublishers.
%
\begin{definition}\label{def:ContextualizedWold}
  A~\emph{specification-annotated \WOLD} (\caWOLD in short) is
    a tuple $\ctxwold=\tuple{\wold, \cwtuple}$ consisting of a \WOLD $\wold=\tuple{\docs,\data,\adoc}$ and an associated family $\cwtuple=\left(\cw_\doc\right)_{\doc\in\docs}$ of subweb specifications.
\end{definition}
In a \caWOLD, each \datapublisher declares their subweb specification that can be used to construct their subweb of interest.
The value of a subweb specification in a \caWOLD is defined as follows:

\begin{definition}
    Let $\ctxwold=\tuple{\wold,\cwtuple}$ be a \caWOLD with $\wold = \tuple{\docs,\data,\adoc}$,
    and $\Theta$ a subweb specification. Then, $\constrEval{\Theta}{\ctxwold}$ denotes the subweb specified by $\Theta$ for $\ctxwold$,
    \[
    \constrEval{\Theta}{\ctxwold} \mydef \bigcup_{(\sselect,b,f) \in \cw}\ \bigcup_{ u\in \sselect(\wold)}
        f\left(\simpl(\adoc(u),\wold)\cup \left(\constrEval{(\Theta_{\adoc(u)})}{\ctxwold} \textbf{ if }b\right), u  \right) ,
    \]
    where $(S\textbf{ if } b)$ equals $S$ if $b$ is true and the empty \WOLD (the unique \WOLD without documents) otherwise.
    The \emph{subweb of interest of a document $d\in \docs$ in \ctxwold} is defined as $\soi(d,\ctxwold) \mydef \simpl(d,\wold) \cup \constrEval{\Theta_d}{\ctxwold}$.

%
%
%
\end{definition}

Since not just the \datapublishers, but also the \queryingagents should be able to specify a subweb of interest,  we naturally obtain the following definition.

\begin{definition}
 A \emph{specification-annotated query} is a tuple $\ctxqry = \tuple{\sparqlex,\cw}$ with $\sparqlex$ a \sparql query and
    $\cw$ a subweb specification.
%
The \emph{evaluation} of \ctxqry in \ctxwold,
denoted \evaluation{\ctxqry}{\ctxwold}, is defined by $
    \evaluation{\ctxqry}{\ctxwold}
\mydef
    \evaluation{\sparqlex}{\constrEval{\Theta}{\ctxwold}}
    $
\end{definition}
Here, we use $\evaluation{P}{W'}$ to denote the evaluation of the \sparql query in the dataset that is the union of all the documents in $W'$ (to be precise, this is the RDF dataset with as \emph{default  graph} the union of all the data in all documents of the subweb, and for each \URI \uri with $\adoc(u)=d$ a \emph{named graph} with name \uri and as triples the data of $d$ \cite{RDF}). 
Of course, we need a mechanism to \emph{find} all those documents, which is what  $\cw$ will provide.



In the next section,  we propose a concrete \sparql-based instantiation of the theoretical framework presented here and illustrate our use case in that setting. Afterwards, we will formally compare our proposal to existing  \ltqp approaches.

%% file: syntax-new.tex
\section{Expressing Subweb Specifications}\label{sec:syntax}
\newcommand\ssl\swsl

In this section, we propose a~syntax for subweb specifications
(as formalized in \cref{sec:Formalization}), named the Subweb Specification
Language~(\ssl), inspired by \ldql and~\sparql.
In order to lower the entry
barrier of this syntax to existing \sparql engine  implementations, we
deliberately base this syntax upon the \sparql grammar.  This enables
implementations to reuse (parts of) existing \sparql query parsers and
evaluators.

The grammar below represents the \ssl syntax in Extended Backus--Naur form (EBNF) with start~symbol~$\bnfpn{start}$.  The specifications begin with the $\bnfts{FOLLOW}$ keyword, followed by a $\bnfpn{sources}$ clause, an
\emph{optional} $\bnfts{WITH SUBWEBS}$ keyword, and an \emph{optional}
$\bnfpn{filter}$ clause.
\begin{bnf*}
  \bnfprod{start}{
    \bnfts{FOLLOW} \bnfsp \bnfpn{sources} \bnfsp [\bnfts{WITH SUBWEBS}] \bnfsp [\bnfpn{filter}]
  }\\
  \bnfprod{sources}{
    \bnfpn{variables} \bnfsp \bnfts{\{} \bnfsp \bnfpn{GroupGraphPattern} \bnfsp \bnfts{\}} \bnfsp [\bnfpn{recurse}]
  }\\
  \bnfprod{variables}{
    \bnfts{?}
    \bnfpn{VARNAME}
    \bnfor
    \bnfts{?}
    \bnfpn{VARNAME} \bnfsp \bnfpn{variables}
  }\\
  \bnfprod{recurse}{
    \bnfts{RECURSE} \bnfsp [\bnfpn{INTEGER}]
  }\\
  \bnfprod{filter}{
    \bnfts{INCLUDE} \bnfsp \bnfpn{ConstructTemplate} \bnfsp [\bnfts{WHERE} \bnfsp \bnfts{\{} \bnfsp \bnfpn{GroupGraphPattern} \bnfsp \bnfts{\}}]
  }\\
\end{bnf*}



Intuitively, a full \ssl expression corresponds to a single subweb specification
tuple $(\sselect,b,f)$ where the $\bnfpn{sources}$ clause correspond to
the source selection function $\sselect$, the keyword $\bnfts{WITH SUBWEBS}$
corresponds to the Boolean $b$, and the $\bnfpn{filter}$ clause corresponds
to the filter function $f$.  We explain each of these parts in more detail
hereafter.

\paragraph{Selection of Sources}
The $\bnfpn{sources}$ will be evaluated in the context of a set $S$ of seed documents.
For subweb specifications provided to the query processor, this set of seeds will be given explicitly, whereas for subweb specifications found in a  document, the set $S$ is comprised of the \URI of that document.
A $\bnfpn{sources}$ clause begins with a list of \sparql variables,
followed by a source extraction expression defined as \sparql's
$\bnfpn{GroupGraphPattern}$ clause. The output is a set of bindings of the given variables, indicating
\URIs whose documents are to be included.
%
%
%
For instance, when evaluating the
expression 
$\bnfts{?}v_1 \bnfsk \bnfts{?}v_n \bnfsp \bnfts{\{} \bnfsp G \bnfsp \bnfts{\}}$
in a \WOLD $\wold$ with seed set $S$, the resulting source selection is
\[\sselect(\wold) = \bigcup_{\uri \in S}  \{  \mu(v_i)  \in \uris  \mid 1\leq i\leq n \land   \mu \in \GGPEval{G}{\data(\adoc(\uri)) }\}, \]
where $\GGPEval{G}{DS}$ is
the evaluation of the GroupGraphPattern $G$ on a dataset $DS$, i.e., a set of bindings $\mu$ (mappings from variables to $\uris\cup\blanks\cup\literals$).

\paragraph{Recurring Source Selection}
A $\bnfpn{sources}$ clause may have at the end an optional $\bnfpn{recurse}$
clause. If $\bnfts{RECURSE}$ is not used in a specification, then this
latter will only apply to the document in which it is defined;
else, the specification will apply to that document,
and all output \URIs, taken as seed (recursively).
In other words, the $\bnfpn{sources}$ clause will be applied to all
documents that are obtained when following a
chain of one or more links using the specification.  The
$\bnfpn{recurse}$ clause has an optional nonnegative integer parameter,
which indicates the maximum recursion \emph{depth}.
A depth of $0$ is equivalent to not defining the $\bnfpn{recurse}$ clause.
A depth of $m$ means that all documents that
are obtained when following a link path of length $m$ from the seeds are considered.
This recursion capability calls for the need to express \emph{the current
document's \URI}. To achieve this, \ssl syntax reuses \sparql's relative \IRI
capability.  Concretely, every time an \ssl specification is applied
on a document, the document's \URI will be set as base \IRI to the \ssl
specification, so that relative \IRIs can be resolved upon this \IRI.

\ignore{
A BIT TOO DETAILED???
To connect it to the semantics, consider the following example on
$\bnfpn{sources}$ expression with $\bnfpn{recurse}$ clause:
\[\bnfts{?}v_1 \bnfsk \bnfts{?}v_n \bnfsp \bnfts{\{} \bnfsp G \bnfsp \bnfts{\}} \bnfsp \bnfts{RECURSE}\]

Given a \WOLD $\wold = \tuple{\docs,\data,\adoc}$, and \URI \uri that indicates
the \URI of the current document where this specification is defined, its
evaluation yields a $\sselect$ function that is defined in terms of
another function $links(\wold,m)$ such that:
\[links(\wold, m) = \bigcup_{i \in \{1, \ldots , n\}}  \{ v \in \uris \mid \text{there exists a mapping } \mu \in \GGPEval{G}{S} \text{ with } \mu(v_i) = v\}\]
where $S = \data(\adoc(\uri))$ for $m = 0$, while $S = \bigcup_{\uri' \in links(\wold, m-1)} \data(\adoc(\uri'))$ for $m > 0$.

Now, $\sselect$ can be defined as $\sselect(\wold) = \bigcup_{i \in \{0, \ldots \}} links(\wold, i)$. Of course, when a limit $M$ is given in the $\bnfpn{recurse}$
clause, then $\sselect(\wold) = \bigcup_{i \in \{0, \ldots , M\}} links(\wold, i)$.
}

\paragraph{Inclusion of Subwebs of Selected Sources} This is determined by the
optional keyword $\bnfts{WITH SUBWEBS}$. Thus, if an \ssl specification has
the $\bnfts{WITH SUBWEBS}$ option, this is equivalent to a subweb specification
tuple with $b$ is $\ltrue$. Otherwise, $b$ is $\lfalse$.

\paragraph{Document Filtering}
The $\bnfpn{filter}$ clause is an optional clause indicating
that only certain parts of the document are considered.
Without this
clause, the entire document is included. The $\bnfpn{filter}$ clause is
similar to \sparql's $\bnfpn{ContructQuery}$ clause. It exists in
\emph{compact} or \emph{extended} forms; in the latter, filtering constraints can be added via $\bnfts{WHERE}$ keyword.

Concretely, the extended form is defined by the \sparql's
$\bnfpn{ConstructTemplate}$ and $\bnfpn{GroupGraphPattern}$ productions.  The
$\bnfpn{ConstructTemplate}$ acts as a template of triples to accept,
while the $\bnfpn{GroupGraphPattern}$ imposes conditions to do so.
%
%
%
%
It is also possible that in the bodies of the $\bnfpn{GroupGraphPattern}$ and
$\bnfpn{ConstructTemplate}$ there are variables that are mentioned in the
$\bnfpn{GroupGraphPattern}$ of $\bnfpn{sources}$ clause.  This implies that they
should be instantiated according to the result of the first
$\bnfpn{GroupGraphPattern}$.

The compact form is defined by $\bnfpn{ConstructTemplate}$, which acts
as syntactical sugar to the extended with an empty $\bnfpn{GroupGraphPattern}$. Thus, to define $\bnfpn{filter}$ clause's semantics, we
only need the extended form. To illustrate this, consider an expression
%
\[\bnfts{FOLLOW}\bnfsp \bnfts{?}v_1\bnfsp \bnfts{\{}\bnfsp G_1\bnfsp \bnfts{\}}\bnfsp\bnfts{INCLUDE} \bnfsp C \bnfsp \bnfts{WHERE} \bnfsp \bnfts{\{} \bnfsp G_2 \bnfsp \bnfts{\}}\]
We already saw that when evaluated in context \uri, this induces a source selector selecting those $v$ such that $\mu_1(\bnfts{?}v_1)=v$, for some $\mu_1 \in \GGPEval{G_1}{\data(\adoc(\uri))}$.
The associated filter is
%
\[f(S,v) = \bigcup_{\mu_1 \in \GGPEval{G_1}{\data(\adoc(\uri))}\mid \mu_1(\bnfts{?}v_1)=v}  \{ t\in S \mid
t \in \GGPEval{\mu_2(\mu_1(C))}{S}
\text{ for some } \mu_2 \in \GGPEval{\mu_1(G_2)}{S} \}\]


\paragraph{Expressing Document Subwebs}
In this work, we assume that each published document can link to its own
context where they indicate the documents they consider relevant using an \ssl
subweb specification.  For illustration, we consider the predicate
$\bnftd{ex:hasSpecification}$ that is attached to the current document.
An $\bnftd{ex:Specification}$ is a~resource that contains at least a~value for
$\bnftd{ex:scope}$, pointing to one or more \ssl strings.  This resource can
also contain metadata about the subweb specification.

\paragraph{Application to the Use Case}
\cref{lst:FollowFriendsScl} shows a part of Uma's profile where she exposes
an \ssl subweb specification to indicate that her friends can express
information about themselves.  This specification states that all
\url{foaf:knows} links from Uma should be followed, and that from those
followed documents, only information about that friend should be included.
By $\bnfts{WITH SUBWEBS}$, she indicates that her friends' subwebs must be
included in her subweb.  Then, Ann can express in her subweb specification
(\cref{lst:FollowFriendsAnnScl}) that she trusts documents pointed to by
\url{foaf:isPrimaryTopicOf} links about triples over the topic she
indicates.  With these subweb specifications, \cref{qry:Friends} produces only
Rows~1--3 of~\cref{tbl:Results}.
However,
we still include the undesired\footnote{Assuming Uma choosing a picture for Bob ``overrides'' the picture Bob has chosen for himself.} profile picture from Bob in our results (Row~3).
Extending the notion of filter to also allow this is left for future work.

\begin{figure}[tb]
%
%
\begin{minipage}[t]{.49\linewidth}
\renewcommand\codefont{\ttfamily\fontsize{7}{8}\selectfont}
\begin{lstlisting}[language=SPARQL,caption=Subweb Specification of https://uma.ex/,label=lst:FollowFriendsScl]
<https://uma.ex/#me> ex:hasSpecification <#spec1>.
<#spec1> ex:appliesTo <https://uma.ex/>;
         ex:scope """
      	  FOLLOW ?friend WITH SUBWEBS {
            <https://uma.ex/#me> foaf:knows ?friend.
          } INCLUDE { ?friend ?p ?o. }
         """^^ex:SWSL.
\end{lstlisting}
\end{minipage}
\begin{minipage}[t]{.49\linewidth}
\renewcommand\codefont{\ttfamily\fontsize{7}{8}\selectfont}
\begin{lstlisting}[
  label=lst:FollowFriendsAnnScl,
  caption= Subweb Specification of https://ann.ex/,
  language=SPARQL,
]
<https://ann.ex/#me> ex:hasSpecification <#spec2>.
<#spec2> ex:appliesTo <https://ann.ex/>;
         ex:scope """
      	  FOLLOW ?page {
            ?topic foaf:isPrimaryTopicOf ?page.
          } INCLUDE { ?topic ?p ?o. }
         """^^ex:SWSL.
\end{lstlisting}
\end{minipage}
\end{figure}

%% file: ldql-requirements.tex
\section{Power and Limitations of Existing \ltqp Approaches}
\label{sec:ldqlComparison}\label{sec:ldql}
Since \ldql is a powerful link traversal formalism that has been shown to subsume other approaches such as reachability-based querying \cite{reachability_semantics},
this raises the question: to what extent can \ldql in itself achieve the requirements set out in \cref{sec:requirements}?
In the current section we formally investigate this, after introducing some preliminaries on \ldql.

\newcommand\lpel{\m{\mathit{lpe}}}
\newcommand\lpeEx\lpel
\newcommand\linkpattern{\m{\mathit{lp}}}
\newcommand\lpetuple[1]{\tuple{#1}}
\newcommand\lpeEvalCtx[2]{\lpeEval{#1}{\enc(\ctxwold)}{#2}}

\subsection{Preliminaries: \ldql}
\ldql is a querying language for linked data. Its most powerful aspect is the navigational language it uses for identifying a subweb of the given \WOLD.
%
The most basic block that constitutes \ldql's navigational language is a \emph{link pattern}: a tuple in
\[(\uris \cup \brak{\_,+})\times(\uris \cup \brak{\_,+})\times(\uris \cup \literals \cup \brak{\_,+}).\]
Intuitively, a link pattern requires a context uri $\uctx$, then evaluates to a set of \URIs (the links to follow) by matching the link pattern against the triples in the document that $\uctx$ is authoritative for.
Formally, we say that a link pattern $\linkpattern=\lpetuple{\ell_1,\ell_2,\ell_3}$
\emph{matches} a triple $\triple{x_1}{x_2}{x_3}$
with result $u$ in the context of a \URI $\uctx$  if the following two points hold:
\begin{enumerate}
  \item
 there exists $i\in \brak{1,2,3}$ such that $\ell_i=\_$ and $x_i=u$, and
 \item  for every $i\in \brak{1,2,3}$ either $\ell_i=x_i$, or $\ell_i=+$ and $x_i=\uctx$, or $\ell_i=\_$.
 \end{enumerate}

Link patterns are used to build \emph{link path expressions} (\lpes) 
with the following syntax:
\begin{align*}\lpel& := \varepsilon \mid \linkpattern \mid \lpel/\lpel \mid \lpel\lpeOr \lpel \mid \lpel^* \mid [\lpel] 
\end{align*}
where $\linkpattern$ is a link pattern
.
In a given \WOLD \wold, the value of a link path expression $\lpel$ in context \URI \uri (denoted $\lpeEval{\lpel}{\wold}{\uri}$) is a set of \URIs
as given in \cref{table:sem2}.

\begin{table}[t]
\centering
 \caption{Value of link path expressions} \label{table:sem2}
\begin{tabular}{l|l}
\toprule
    $\lpeEx$&$\lpeEval{\lpeEx}{\wold}{u}$\\
 \midrule
    $\epsilon$&$\{u\}$\\
 $\linkpattern$ & $\{u'\mid \linkpattern \text{ matches with $t$ $($with result $u'$ in context $u$ for some $t\in\data(\adoc(u))$}\}$\\
 $\lpel_1/\lpel_2$& \brak{v \mid v\in \lpeEval{\lpel_2}{\wold}{u'} \text{ and } u'\in \lpeEval{\lpel_1}{\wold}{u}}\\
 $\lpel_1\lpeOr \lpel_2$&$\lpeEval{\lpel_1}{\wold}{u} \cup \lpeEval{\lpel_2}{\wold}{u}$\\
 $\lpel^*$&$\{u\} \cup \lpeEval{\lpel}{\wold}{u} \cup \lpeEval{\lpel/\lpel}{\wold}{u} \cup \lpeEval{\lpel/\lpel/\lpel}{\wold}{u} \cup ...$\\
 $[\lpel]$ &$\brak{u \mid \lpeEval{\lpel}{\wold}{u}\ne \emptyset}$\\
 \bottomrule
\end{tabular}
\end{table}

An \ldql query is a tuple $q = \tuple{\lpel,\sparqlex}$ with $\lpel$ a link path expression and $P$ a \sparql query.
The value of such a query $q$ in a \WOLD \wold with a set of seed \URIs $S$ is
\[\ldqlEval{q}{\wold}{S} := \queryEval{P}{}{W'} \text{ where } W'= \bigcup_{s\in S, u\in  \lpeEval{\lpel}{\wold}{s}} \simpl(\adoc(u),W),
\]
i.e., the query $P$ is evaluated over the (RDF dataset constructed from the) data sources obtained by evaluating the link path expression starting in one of the seeds.

\begin{remark} \citet{ldql} allow one other form of link path expression, where an entire \ldql query is nested in an \lpe; for the purpose
 of this paper, we opt to use a strict separation between query and source selection and omit this last option.\footnote{Notably, this option was also not present in the original work  \cite{ldql-conf}.}
 Additionally, they consider (Boolean) combinations of queries, thereby allowing to use different \lpes for different parts of the expression; we briefly come back to this when discussing scope restriction.
\end{remark}


\subsection{\ldql and the Requirements }

\paragraph{A Declarative Language for Selecting Data Sources}
In \ldql, the link path expressions provide a rich and flexible declarative language for describing source selection.
Here, paths through the linked web are described using a syntax similar to regular expressions. 
For instance, the \ldql expression
\[
\lpetuple{+,\text{\lstinlineProxy{foaf:knows}},\_}/\lpetuple{+,\text{\lstinlineProxy{foaf:knows}},\_}
\]
when evaluated in a given \URI \uri (the context) traverses to \uri's friends $f$ (as explicated by triples of the form \triple{u}{\text{\lstinlineProxy{foaf:knows}}}{f} in $\adoc(\uri)$) and subsequently to their friends $f_2$ (as indicated by triples \triple{f}{\text{\lstinlineProxy{foaf:knows}}}{f_2} in $\adoc(f)$).
The final result contains only such friends $f_2$.
In other words, this example expression identifies the documents of friends of friends of a given person.

\paragraph{Independence of Query and Subweb Specification}
The design philosophy behind \ldql does not start from an independence principle similar to the one proposed here.
That is, in its most general form, \ldql allows intertwining the source selection and the query. For instance, the  \ldql query
\[
\langle \lpeEx_1, \sparqlex_1\rangle\ldqlAND{} \langle \lpeEx_2, \sparqlex_2\rangle
\]
%
expresses the \sparql query $\sparqlex_1\ldqlAND \sparqlex_2$, and on top of that specifies that different parts of the query should be evaluated with respect to different sources: $\sparqlex_1$ should be evaluated in the documents identified by $\lpeEx_1$ and $\sparqlex_2$ in the documents identified by $\lpeEx_2$.
In this sense, \ldql thus violates our principle of independence.
However, independence can easily be achieved in \ldql by only considering \ldql queries of the form
$
 \langle \lpel,\sparqlex\rangle
$
with $\lpel$ a link path expression and $\sparqlex$ a \sparql query.

\paragraph{Scope Restriction of Sources}
The semantics of an \ldql query
$
 \langle \lpeEx,\sparqlex\rangle
 $
is obtained by first evaluating \lpeEx starting from a seed document $\seed$, resulting in a set of \URIs $\lpeEval{\lpeEx}{W}{\seed}$; the \sparql query \sparqlex is then evaluated over the union of the associated documents.
That is, to compute the result of $\langle \lpeEx,\sparqlex\rangle$, for each document $\adoc(\uri)$  with $ \uri \in \lpeEval{\lpeEx}{W}{\seed}$, its entire content is used.
As such, \ldql provides no mechanism for partial inclusion of documents. 
However, while \ldql cannot select \emph{parts of documents}, it \emph{can} be used, as discussed above, to apply source selection strategies only to \emph{parts of queries} and thereby to a certain extent achieve the desired behaviour.
For instance, the query
\[
  \langle \lpeEx_1, (?x, \text{\lstinlineProxy{foaf:knows}},?y) \rangle\ldqlAND{} \langle \lpeEx_2, (?y, \text{\lstinlineProxy{foaf:mbox}}, ?m) \rangle
  \]
 will only use triples with predicate \text{\lstinlineProxy{foaf:knows}} from documents produced by $\lpeEx_1$.
%
However, this sacrifices the independence property, and for complex queries and filters, this is not easy to achieve. 

\newcommand\classS{\m{\mathcal{S}}}
\newcommand\enc{\m{\mathit{enc}}}
\newcommand\lpemeta{\m{e_{\mathit{meta}}}}
\newcommand\id{\m{\mathit{id}}}

\paragraph{Distributed Subweb Specifications}

This now brings us to the main topic of this section:
studying to which extent it is possible in \ldql to distribute the knowledge of how to construct the subweb of interest and as such to \emph{guide} the data consumer towards interesting/relevant documents.
%
%
To answer this question, we will consider a slightly simplified setting, without filters (all filters equal the identity function \id on their first argument) and where the Boolean $b$ in $(\sselect,b,f)$ is always true. I.e., each agent states that they wish to include the complete subweb of interest of all \URIs  identified by $\sselect$.
In this setting, we wonder if data publishers can, instead of publishing their subweb specification \emph{in addition to} their regular data, encode their subweb specification as triples \emph{in} the document (as meta-information), and use \emph{a single} ``meta'' link path expression that interprets these triples for the traversal. This is formalized as follows.


\begin{definition}
    Let $\classS$ be a set of source selectors, $\enc:\classS\to 2^\alltriples$ a function mapping source selectors $\sselect$ onto a set of triples $\enc(\sselect)$, and $\ctxwold=\tuple{\wold,\cwtuple}$ a \caWOLD (with $\wold =\tuple{\docs,\data,\adoc}$) in which each subweb specification is of the form $(\sselect,\ltrue,\id)$ with $\sselect\in \classS$.
    The \emph{encoding of \ctxwold} by $\enc$ is the \WOLD $\enc(\ctxwold)=\tuple{\docs, \data', \adoc}$ with for each $d\in\docs$:
  \[\data'(d) = \data(d) \cup  \bigcup_{\{\sselect\mid (\sselect,\ltrue,\id)\in \cw_\doc\}}\enc(\sselect).\]
\end{definition}
\newcommand\docsof[1]{\m{\mathit{docs}(#1)}}
\begin{definition}
    Let $\classS$ be a set of source selectors, \enc a function $\classS\to 2^\alltriples$,  and $\lpemeta$  an \lpe.
We say that  $(\enc,\lpemeta)$ \emph{captures} $\classS$ if for each \caWOLD $\ctxwold=\tuple{\wold,\cwtuple}$ with $\wold =\tuple{\docs,\data,\adoc}$ and in which subweb specifications only use triples of the form  $(\sselect,\ltrue,\id)$ with $\sselect\in\classS$ and for each \URI $u$,
 \[\docsof{\lpeEval{\lpemeta}{\enc(\ctxwold)}{u}} = \soi(\adoc(u),\ctxwold),\]
 where
 \[ \docsof{S} = \bigcup_{s\in  S} \simpl(\adoc(s),\wold).\]
We will say that \emph{\ldql can capture distribution of functions in \classS} if there exist some \enc and $\lpemeta$ that capture \classS.
%
\end{definition}
What this definition states is that the \lpe \lpemeta, when evaluated in $s$ identifies precisely all \URIs needed to create the subweb of interest of $s$ (including $s$ itself).  In case \ldql captures distribution of functions in a certain class \classS, this means that the knowledge of selectors in \classS can be encoded as ``meta-triples'' in the documents to guide the querying agent towards the relevant data sources. In what follows, we study for some concrete classes whether \ldql can capture distribution.

\newcommand\Scons{\m{\classS_{\mathit{const}}}}
\newcommand\Sp{\m{\classS_{p^*}}}
\newcommand\Sall{\m{\classS_{*}}}

To define the encodings, we will make use of some ``fresh'' \URIs we assume not to occur in any \WOLD. In our theorems, we will make use of some specific sets of source selectors. A source selector $\sselect$ is \emph{constant} if it maps all \WOLDs onto the same set of \URIs, i.e., if $\sigma(W)=\sigma(\wold')$ for all \WOLDs $\wold,\wold'$; the set of all constant source selectors is defined as $\Scons$.
If $p$ and $u$ are \URIs, we define the source selector $\mathit{all}_{p^*,u}$ as follows:
\[\mathit{all}_{p^*,u}: \wold \mapsto \lpeEval{\lpetuple{+,p,\_}^*}{\wold}{\uri}.\]
Intuitively, the function $\mathit{all}_{p ^*,u}$ identifies the set of all $p$s of $p$s of .... of $u$. For instance, by taking $p = \mathit{friend}$, we include all direct or indirect friends of $u$.
For a fixed $p$, we write $\Sp$ for the set of source selectors $\mathit{all}_{p^*,u}$. We write $\Sall$ for the set of all source selectors of the form $\mathit{all}_{p^*,u}$ for any $p$.
The set \Sall allows each data publisher to choose her own strategy for constructing the subweb, e.g., one \datapublisher might include all her $\mathit{friend}^*$s, another her $\mathit{colleague}^*$s and a third one only \URIs explicitly trusted (i.e., their $\mathit{trust}^*$s). 

Our main (in)expressivity  results are then summarized in the following Theorems. 

\begin{theorem}\label{thm:expressivity1}
\ldql captures the distribution of $\Scons$.
\end{theorem}%
\begin{proof}
We will provide explicit pairs of encoding and meta-expression that capture the distribution.

Consider the class $\Scons$; for any (constant) source selector $\sselect$ in this class, let $U$ be the set of sources defined by it. In this case, we take $\enc(\sselect)= \{\triple{a}{a}{u} \mid u \in U\}$ and $\lpemeta = \lpetuple{a,a,\_}^*$ with $a$ a fresh \URI. The link pattern $\lpetuple{a,a,\_}$ in \lpemeta is used to navigate to the \uri, while the star ensures that for each \uri that is found, also their subwebs of interests are included.
\end{proof}
\begin{theorem}\label{thm:expressivity2}
\ldql captures the distribution of $\Sp$.
\end{theorem}%
\begin{proof}
We will provide explicit pairs of encoding and meta-expression that capture the distribution, similar to the proof of~\cref{thm:expressivity1}, but now for the class \Sp.

We can take $\enc(\mathit{all}_{p^*,u})= \{\triple{a}{a}{u}\}$ and $\lpemeta = (\lpetuple{a,a,\_}/\lpetuple{+,p,\_}^*)^*$ with $a$ a fresh \URI.
In this expression \lpemeta, the link pattern $\lpetuple{a,a,\_}$ is used to navigate to the \uri as in the case of $\Scons$. Each of these \URIs is a source whose $p$s of $p$s of... we wish to include; the part $\lpetuple{+,p,\_}^*$ then navigates to all such $p^*$s.
Again, the outermost star ensures that for each \uri that is found, also their subwebs of interests are included.
\end{proof}

Intuitively, \ldql captured the classes $\Scons$ and $\Sp$, because we
wanted the encoding to provide only one piece of information, which is
the source to navigate to.
Hence, the usage of the meta-expression is to decode that source.
In the case of the $\Sall$ distribution, there are two pieces of
information that need to be provided by the encoding function
	\begin{andlist}
		\item the source to navigate to
		\item the property we need to follow from that source
	\end{andlist}.
This can be encoded by some simple \enc function.
Hence, the meta-expression would then need to decode both
the source and the property. However, regardless of how
the \enc function will be defined, there is no generic \lpe
that does this correctly. This is proved in the following theorem.

\begin{figure}[t]
  \centering
  \begin{tikzpicture}[
nodeRec/.style={rectangle, fill=black!10, draw=blue!40!black!60, thick},
edge/.style = {->,> = latex'},
]
  \node[nodeRec] (1) [rectangle split, rectangle split, rectangle split parts=2,align=left] { {$d_1= \adoc(u_1)$} \nodepart{second} {$\data= \{\triple{u_1}{p}{u_2}, \triple{u_1}{q}{u_3}\}$}};
  \node (1d) [above = 0.4cm of 1] {``include $p^*$s of $u_1$''};
  \draw[dashed](1)--(1d);
  \node[nodeRec] (2) [above right= -0.4cm and 1.2cm of 1, rectangle split, rectangle split, rectangle split parts=2] { {$d_2= \adoc(u_2)$} \nodepart{second} {$\data=\emptyset$}};
  \node[nodeRec] (3) [below right= -0.4cm and 1.2cm of 1, rectangle split, rectangle split, rectangle split parts=2] { {$d_3= \adoc(u_3)$} \nodepart{second} {$\data=\emptyset$}};

  \path[->]
  (1) edge               node[above]        {$p$}    (2)
  (1) edge               node[above]        {$q$}    (3)
   ;

\end{tikzpicture}
\caption{Example \WOLD used in \ldql inexpressivity proof.}
\label{fig:ex}
\end{figure}
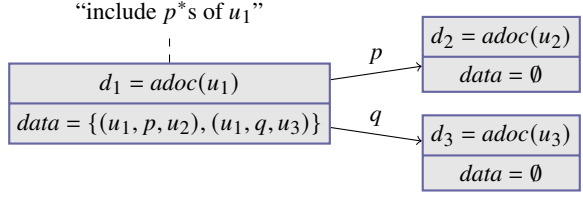
\begin{theorem}\label{thm:inexpressivity}
\ldql does not capture the distribution of $\Sall$.
\end{theorem}%
\begin{proof}
For the sake of contradiction, assume that \ldql captures the
distribution of $\Sall$.  Then, there exists some pair of encoding
and meta-expression that captures \Sall.
Let this pair be $(\enc,\lpemeta)$ with $U$ the set of \URIs
mentioned in \lpemeta.
We can construct a \WOLD (see~\cref{fig:ex}) that uses \URIs
that do not occur in $U$ in which only one document has a non-empty
subweb specification.
As shown in~\cref{fig:ex}, we have two triples $\triple{u_1}{p}{u_2}$
and $\triple{u_1}{q}{u_3}$ in $d_1$ where none of $p, q, u_2$, or $u_3$
is in $U$.
Moreover, $d_1$ has a subweb specification whose source selector is
$\mathit{all}_{p^*,u_1}$.
We can also make sure that neither $u_2$ nor $u_3$ is mentioned in the
triples added by $\enc$.  This can be safely assumed since $\enc$ does
not depend on the triples found in the \wold, rather it only depends on
the source selector.

Now, if $\lpemeta$ is a correct meta-expression,
$\lpeEval{\lpemeta}{\enc(\ctxwold)}{u_1}$ should evaluate to
$\{u_1, u_2\}$.
However, any $\lpemeta$ when evaluated at $u_1$ in our \ctxwold
would either include $u_3$ as a selected source or exclude $u_2$
from the selected sources.
Thus, in the rest of the proof we verify this claim.

In order to verify this claim, we first need to show that given
any \URI $u$ whose authoritative document and subweb specification
are empty in some \ctxwold, then for every \lpe $e$, we have that
\[\lpeEval{e}{\enc(\ctxwold)}{u} \subseteq \{u\}\]
Intuitively, if we have an empty document without any associated
subweb specification, then any \lpe evaluated at the source of
this document would not result in any sources except for the source
of that document, which the context \URI.
This can be shown by induction on the shape of the \lpe $e$ as follows:
\begin{itemize}
\item If $e$ is $\epsilon$ or $e$ is $[\lpel]$, it is clear that
$\lpeEvalCtx{e}{\uri} \subseteq \{\uri\}$.
\item If $e$ is $\linkpattern$, we have $\lpeEvalCtx{e}{\uri} = \emptyset$
since the document is empty and $\uri$ has no subweb specification.
\item If $e$ is $\lpel_1/\lpel_2$, it follows by induction that
$\lpeEvalCtx{\lpel_1}{\uri} \subseteq \{\uri\}$.
In case $\lpeEvalCtx{\lpel_1}{\uri}$ is empty, then
$\lpeEvalCtx{e}{\uri} = \emptyset \subseteq \{\uri\}$.
Otherwise, $\lpeEvalCtx{\lpel_1}{\uri} = \{\uri\}$.
In that case, it also follows by induction that
$\lpeEvalCtx{\lpel_2}{\uri} = \lpeEvalCtx{e}{\uri} \subseteq \{\uri\}$.
\item If $e$ is $\lpel_1 \lpeOr \lpel_2$, it follows by induction that
both $\lpeEvalCtx{\lpel_1}{\uri}$ and $\lpeEvalCtx{\lpel_2}{\uri}$ are
subsets of $\{\uri\}$.
Hence, $\lpeEvalCtx{e}{\uri} = \lpeEvalCtx{\lpel_1}{\uri} \cup
\lpeEvalCtx{\lpel_2}{\uri} \subseteq \{\uri\}$.
\item If $e$ is $\lpel^*$, then $\lpeEvalCtx{e}{\uri} \subseteq \{\uri\}$
follows by induction.
\end{itemize}

Applying this to our example, we see that for any \lpe $e$,
it follows that $\lpeEvalCtx{e}{u_2} \subseteq \{u_2\}$ and
$\lpeEvalCtx{e}{u_3} \subseteq \{u_3\}$.
Now we need to show that for every \lpe $e$,
\[u_2\in \lpeEval{e}{\enc(\ctxwold)}{u_2} \text{ if and only if }
u_3 \in \lpeEval{e}{\enc(\ctxwold)}{u_3}\]

This can also be verified by induction on the shape of the \lpe $e$,
however, it is easily shown from the previous induction since in all the
cases where $\lpeEvalCtx{e}{\uri}$ did not turn out empty, the result
did not depend on the shape of the link patterns used in the expression
rather it originally came from the \lpe $\epsilon$ which is evaluated
similarly in $u_2$ and $u_3$.

Now, in our example \ctxwold, we show that for every \lpe \lpemeta that
does not mention any of the \URIs $p$, $q$, $u_2$, or $u_3$, we have that
\[u_2\in \lpeEvalCtx{\lpemeta}{u_1} \text{ if and only if }
u_3 \in \lpeEvalCtx{\lpemeta}{u_1}\]
We verify this claim by induction on the possible shapes of $\lpemeta$
as follows:
\begin{itemize}
\item If $\lpemeta$ is $\epsilon$ and $\lpemeta$ is $[\lpel]$, it is clear
that $\lpeEvalCtx{\lpemeta}{u_1} \subseteq \{u_1\}$.
\item If $\lpemeta$ is $\linkpattern$, we have three cases to analyze:
\begin{itemize}
    \item $\linkpattern$ matching a triple that is added by $\enc$.
    As mentioned, none of the triples added by
    $\enc(\mathit{all}_{p^*,u_1})$ to $d_1$ mentions $u_2$ or $u_3$.
    Hence, it is clear that $u_2 \not \in \lpeEvalCtx{\linkpattern}{u_1}$
    and $u_3 \not \in \lpeEvalCtx{\linkpattern}{u_1}$.
    \item $\linkpattern$ of the form $\lpetuple{l_1,l_2,l_3}$ matching the
    triple $\triple{u_1}{p}{u_2}$ in $d_1$.
    Since \lpemeta mentions neither $p$ nor $u_2$, we must have
    $l_1 \in \{u_1, +, \_ \}$ and $l_2 = l_3 = \_$ in order to match this
    triple.
    Clearly, each of the three possible link patterns matches the triple
    $\triple{u_1}{q}{u_3}$ as well.
    Thus, $\{u_2, u_3\} \subseteq \lpeEvalCtx{\linkpattern}{u_1}$.
    \item for any other $\linkpattern$, we have
    $\lpeEvalCtx{\linkpattern}{u_1} = \emptyset$ since there are no other
    triples in the document.
\end{itemize}
\item If $\lpemeta$ is $\lpel_1/\lpel_2$,
it follows by induction that $u_2 \in \lpeEvalCtx{\lpel_1}{u_1}$
if and only if $u_3 \in \lpeEvalCtx{\lpel_1}{u_1}$.
Accordingly, we have three cases to analyze:
\begin{itemize}
    \item neither $u_1$, $u_2$ nor $u_3$ belongs to
    $\lpeEval{\lpel_1}{\enc(\ctxwold)}{u_1}$.
    In this case, neither $u_2$ nor $u_3$ belongs to
    $\lpeEval{\lpemeta}{\enc(\ctxwold)}{u_1}$.
    \item $u_1 \in \lpeEval{\lpel_1}{\enc(\ctxwold)}{u_1}$.
    By induction, we have that $u_2 \in \lpeEvalCtx{\lpel_2}{u_1}$
    if and only if $u_3 \in \lpeEvalCtx{\lpel_2}{u_1}$.
    \item $\{u_2, u_3\} \subseteq \lpeEvalCtx{\lpel_1}{u_1}$.
    Thus, the expression $\lpel_2$ will be evaluated once at $u_2$
    and once at $u_3$.
    By our previous discussion, we can verify that
    $u_2 \in \lpeEval{\lpel_2}{\enc(\ctxwold)}{u_2}$ if and only if
    $u_3 \in \lpeEval{\lpel_2}{\enc(\ctxwold)}{u_3}$.
\end{itemize}
Thus, in all three cases, $u_2 \in \lpeEvalCtx{\lpemeta}{u_1}$
if and only if $u_3 \in \lpeEvalCtx{\lpemeta}{u_2}$.
\item If $\lpemeta$ is $\lpel_1 \lpeOr \lpel_2$ or $\lpemeta$ is $\lpel^*$,
it follows by induction.
\end{itemize}
In all cases, we have showed that
$u_2 \in \lpeEval{\lpemeta}{\enc(\ctxwold)}{u_1}$ if and
only if $u_3 \in \lpeEval{\lpemeta}{\enc(\ctxwold)}{u_2}$.
Hence, a contradiction.
\end{proof}

%% file: experiments.tex
\section{Experiments}
\label{sec:experiments}
In this section, we present an empirical evaluation of our proposed formalism.
The aim of these experiments is to evaluate what the possible gain is when using subweb specifications.
Since for the moment, there are no data publishers publishing their subweb specifications, we created a virtual linked web,  in which agents do publish such specifications, of which we assume they represent which data sources they trust.
On this virtual linked web, we will compare querying the annotated web in our new semantics to querying it under existing reachability semantics.
We expect to observe the benefits of using subweb specifications in terms of completeness of the query results, data quality, and performance.
As far as data quality and completeness go, we assume here that, since each agent publishes their subweb specifications, the results obtained by querying the subweb-annotated \WOLD are considered the ``gold standard'' (they are the results we wish to obtain since they take every agent's specifications into account).
What we want to evaluate is how well existing link traversal can approximate this gold standard: whether query results are missing (incomplete) or whether spurious results are obtained (indicating the use of lower-quality data).
We will also evaluate performance: the time needed to traverse the web, as well as the number of traversals required.

Our experiments are run on a virtual linked web constructed from the LDBC SNB (Social Network Benchmark) dataset~\cite{ldbc_snb_interactive}.
The resulting dataset is then fragmented into interlinked datasources resulting in a virtual web of around 301000 datasources with 1.5GB of data.
The schema of the virtual linked web is illustrated in Figure \ref{fig:exp-schema}.
To each document type, we attach a subweb specification with the intended meaning that it identifies other data on the web it trusts.
We use simple specifications that can naturally arise in real-life scenarios.  Most of the defined specifications use filters to only include data that is related to the respective source and ignore other triples that provide information about others.
For instance, a person trusts (includes data from) other persons they know,
their university, the city where they live, and the company in which they work.
The experiments are then run in the annotated virtual linked web.
The subweb specifications used in the annotated linked web
are listed in \ref{sec:subwebs-annot}.

We compare our semantics with the most common reachability semantics criteria \cite{reachability_semantics}:
\cAll, \cNone, and \cMatch.
Our expectation is that \cAll, which follows all links it finds, will be too slow for practical purposes, and generate too many spurious results.
For \cNone, which simply follows no links (and thus only queries the given seed documents), we expect extremely fast querying, but almost no query results, due to the fragmentation of the data.
For \cMatch, which follows all links that somehow ``match'' the query in question will miss certain query results, but can also generate some spurious results. Indeed, the traversal performed by \cMatch is only informed by the \queryingagent, not by the individual users on the web.
The \emph{performance} (in terms of runtime and number of links traversed) of \cMatch compared to subweb-annotated querying is hard to predict.
Clearly, the behavior of \swsl itself differs with respect to the
filters used in the specifications, which will be elaborated in two of
the experimented queries.

\begin{figure}
	\centering
	\includegraphics[width=\textwidth]{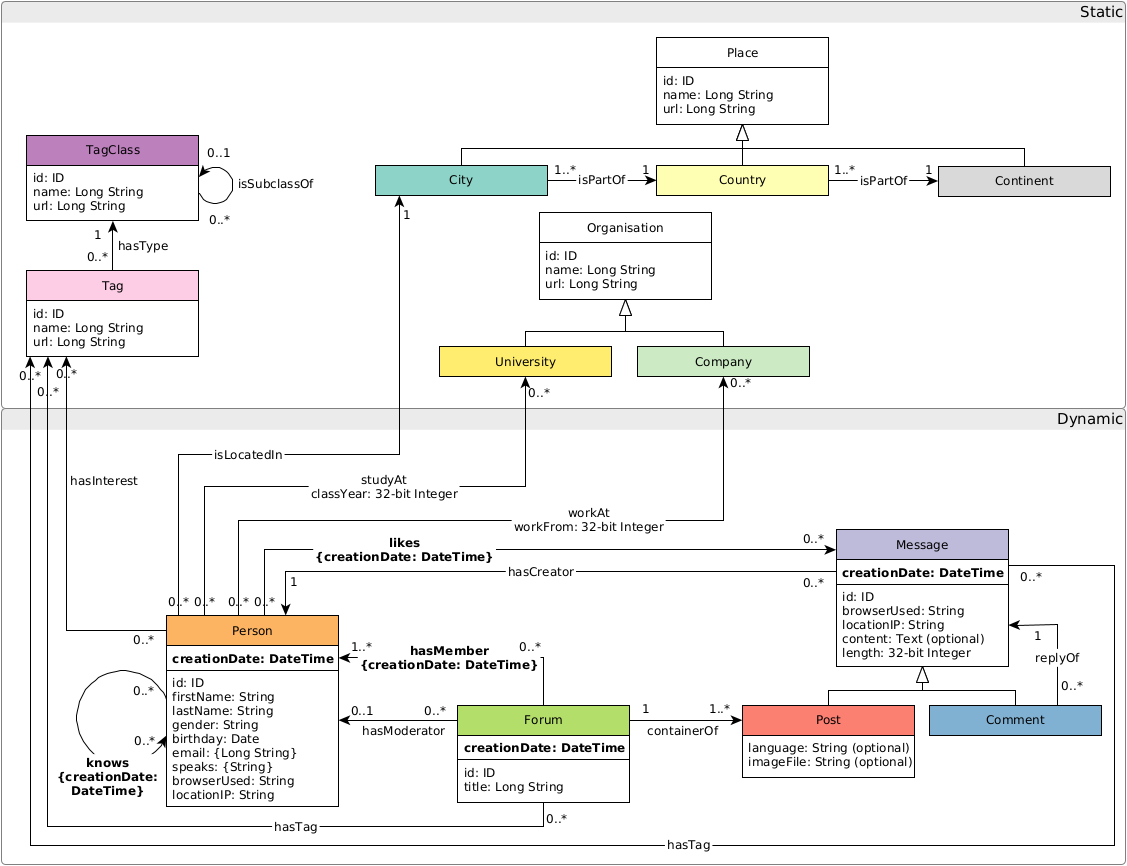}
	\caption{Schema of the virtual linked web used in the experiments (from \url{https://www.npmjs.com/package/ldbc-snb-decentralized}).
}
	\label{fig:exp-schema}
\end{figure}

The virtual web is hosted locally using a (community\footnote{https://communitysolidserver.github.io/CommunitySolidServer/docs/}) solid server.
Each document on the web provides information in the form of a dump file, that is, the \rdf dataset which should be downloaded entirely prior to querying.
We use the comunica query engine \cite{taelman_iswc_2018} as a \sparql querying engine; the engine uses a caching mechanism that avoids downloading the same file more than once.
The experiments are conducted on a personal computer with 16GB of memory and an Intel Core i7 with 2.6 GHz and 6 cores running macOS 12.6.

\subsection{Queries}

For the evaluation, we used four different queries.
The first query (\cref{eval:q1}) searches for persons' information,
it retrieves for each person, the city where they are located along
with the country and the continent of this city.
The second query (\cref{eval:q2}) matches persons working in the
same company, in which those persons' are retrieved along with
the name of the company.
In the third query (\cref{eval:q3}), we test forums members' interests;
we want to list members (and moderators) of a forum
that have no interest (tag) in common with the forum.
The last query (\cref{eval:q4}) is used to retrieve all persons
having an interaction between them; an interaction here is simply
a person liking a comment of another; we also list the city where
the person who performed the like lives in.  For the first two queries,
we evaluated our semantics using two different subweb specifications.
We refer to the main subweb specification used in all queries as \ssl.
The other subweb specification used in evaluating the first query will
be referred to as {\swslOne}, while \swslTwo is the other subweb
specification used for the second query.  The main difference between
\swslOne and \swslTwo and \swsl is the filters of the subweb
specifications attached to person-type documents.  Simply, we allow
more triples to be included in the subwebs of person-type documents
in the case of \swslOne and \swslTwo than that of \swsl, so the
filters used in \swsl are more restricted
(see~\ref{sec:subwebs-annot} for details).

\begin{minipage}{0.48\textwidth}
\begin{test}{eval:q1}{Person's location}
\begin{lstlisting}[language=SPARQL,xleftmargin=0.15\textwidth]
SELECT ?person ?country ?cont
WHERE {
  ?person voc:isLocatedIn ?city.
  ?city voc:isPartOf ?country.
  ?country voc:isPartOf ?cont.
}
\end{lstlisting}
\end{test}

\hfill\hrule\hfill

\begin{test}{eval:q2}{Same company}
\begin{lstlisting}[language=SPARQL,xleftmargin=0.15\textwidth]
SELECT ?person1 ?person2 ?namecomp
WHERE {
  ?person1 voc:workAt ?c1.
  ?c1 voc:hasOrganisation ?comp1.

  ?person2 voc:workAt ?c2.
  ?c2 voc:hasOrganisation ?comp2.

  FILTER(?person1 != ?person2)

  ?comp1 foaf:name ?namecomp.
  ?comp2 foaf:name ?namecomp.
}
\end{lstlisting}
\end{test}
\end{minipage}\hfill\vline\hfill
\begin{minipage}{0.48\textwidth}
\begin{test}{eval:q3}{Forum member interests}
\begin{lstlisting}[language=SPARQL,xleftmargin=0.15\textwidth]
SELECT ?forum ?creator
WHERE {
  {?forum voc:hasModerator ?creator.}
  UNION
  {?forum voc:hasMember ?creatorB.
   ?creatorB voc:hasPerson ?creator.}
  ?creator voc:hasInterest ?interest.

  FILTER (NOT EXISTS {
    SELECT ?tagForum WHERE {
      ?forum voc:hasTag ?tagForum.
      FILTER (
        bound(?interest) &&
        ?tagForum = ?interest )}})
}
\end{lstlisting}
\end{test}

\hfill\hrule\hfill

\begin{test}{eval:q4}{Interaction}
\begin{lstlisting}[language=SPARQL,xleftmargin=0.15\textwidth]
SELECT ?person ?creator ?city
WHERE {
  ?person voc:isLocatedIn ?city.
  ?person voc:likes ?message.
  ?message voc:hasComment ?comm.
  ?comm voc:hasCreator ?creator.
}
\end{lstlisting}
\end{test}
\end{minipage}
\vspace*{0.5cm}

Each query is executed 12 times
with a different (random) seed \URI for each run,
and the average is reported.
We use the same set of seeds for all strategies
to avoid bias related to seeds selection,
for instance,
forums can have different number of members
and persons can have
different numbers of posts, comments, or friends.

\subsection{Results}

The results are displayed in Tables \ref{table:experiments-res1}--\ref{table:experiments-res4} (one table for each query).
In all the experiments, \cAll and \cNone behave as predicted,
that is,
the engine traverses a large number of links when using \cAll,
while no links are traversed when \cNone is used.


The first query (see \cref{table:experiments-res1})
shows the best traversal performance for \cMatch
to the detriment of query results.
Using \swslOne, the engine was able to yield on average 16 results,
while \cMatch was only able to find one of these results.
In fact, for each correct result retrieved by \swslOne,
a person entity is involved, and since no triple pattern in the query
links to other persons, \cMatch only has the seed document to generate
a result in response to the query.
On the other hand, \swslOne was able to traverse to other persons
using the subweb specifications of the seed document.
It may seem weird to have no results at all from \ssl,
but this makes sense since the defined filters are quite strong.
In order to be able to include the triple that links the country to
its continent, which is of the form
\verb|(_:country voc:isPartOf _:continent)|, in the subweb of
the current person, this can only be done
via the \URI of the city where this person is located in.
Although this triple is included in the subweb of the city,
the filter does not allow the necessary triple to be included
in the subweb of the person as it only allows for triples whose
subject is the city itself.

The second query shows similar behavior for \ssl (see Table \ref{table:experiments-res2}).
The engine didn't yield any results using it.  This is due to the
structure of the data inside documents and the defined filters.
A company, where a person works, is not declared using one triple,
instead, a blank node is used as an intermediate entity.
That is, we use two triples as in
\verb|(_:person voc:workAt _:blankWork)| and
\verb|(_:blankWork voc:hasOrganisation _:company)|.
This makes the company inaccessible from the subweb of the
data publisher, the reason for this is the filter used by
the subweb specification published at the person's data set.
The filter only allows triples having the person in the
subject position, this will only include the triple
\verb|(_:person voc:workAt _:blankWork)|
which doesn't refer to the actual company.
As for \swslTwo, on average eleven results are obtained.
This change is due to the inclusion of triples whose property
is either \verb|voc:hasOrganisation| or  \verb|foaf:name| from
the subweb of person.
What this shows is that some care is required when using
excessive filtering.
As for \cMatch, the engine did not yield any results since no
triple pattern links for other persons is mentioned in the query.

In the third query (see Table \ref{table:experiments-res3}), \cAll, \cMatch, and \ssl generated the same results.
Nevertheless, \cMatch has to do more link traversal to achieve the querying process,
it has also fetched more triples for this.
This is mainly due to the fact that \cMatch does the traversal
by using information from the query,
the more triples patterns the query has,
the more links the engine will have to traverse.

In the last query, \cMatch produces on average 18 results that
seem to be missing from \ssl.  This is not true, since all the extra
tuples retrieved by \cMatch are duplicates.  Thus, both \cMatch
and \ssl obtained the same \emph{set} of results, but \cMatch did it in a
more efficient manner.  The reason for this behavior of \ssl is that defined
subweb specifications for persons include a lot of irrelevant data to the
query, which are pruned earlier with \cMatch.


\begin{table}
\centering
\begin{tabular}{lrrrrr} 
	 \toprule
	Approach & \# traversed links & traversal time & \# triples & query evaluation time & \# results \\
	 \midrule
	 \ssl & 8187 & 88 s & 5124 & 1 s & 0\\
	 \midrule
 	  \swslOne & 8187 & 98 s & 5550 & 1 s & 16 \\
	 \midrule
	 \cMatch & 4 & 0 s & 703 & 1 s & 1\\
	 \midrule
	 \cAll & 29969 & 474 s & 698126 & 11 s & 1921\\
	 \midrule
	 \cNone & 0 & 0 s & 343 & 0 s & 0\\
	 \bottomrule
\end{tabular}
\caption{Performance results for (\cref{eval:q1}).  In evaluating \cAll for
this query, the experiment did not yield any results for one of the twelve
tested seed sources.  The reason was that the number of triples collected
from traversing the links phase was so huge, which made the \sparql query
engine crash.  Hence, the number of triples, the query evaluation time,
and the number of results for \cAll are the averages of the other eleven runs.}
\label{table:experiments-res1}
\end{table}

\begin{table}
\centering
\begin{tabular}{lrrrrr} 
	\toprule
	Approach & \# traversed links & traversal time & \# triples & query evaluation time & \# results \\
	\midrule
	\ssl & 5294 & 53 s & 3195 & 1 s & 0\\
	\midrule
	\swslTwo & 5294 & 55 s & 4703 & 1 s & 11\\
	\midrule
	\cMatch & 5 & 0 s & 556 & 1 s & 0\\
	\midrule
	\cAll & 27251 & 423 s & 597198 & 362 s & 16372\\
	\midrule
	\cNone & 0 & 0 s & 279 & 0 s & 0\\
	\bottomrule
\end{tabular}
\caption{Performance results for (\cref{eval:q2}).  In evaluating \cAll
for this query, the experiment did not yield any results for three of the
twelve tested seed sources.  The reason was the same reason as the one
mentioned in \cref{table:experiments-res1}.  Hence, the number of triples, the query
evaluation time, and the number of results for \cAll are the averages of
the other nine runs.}
\label{table:experiments-res2}
\end{table}

\begin{table}
\centering
\begin{tabular}{lrrrrr} 
	\toprule
	Approach & \# traversed links & traversal time & \# triples & query evaluation time & \# results \\
	\midrule
	\ssl & 14 & 1 s & 4567 & 1 s & 134\\
	\midrule
	\cMatch & 641 & 50 s & 221487 & 29 s & 134\\
	\midrule
	\cAll & 32701 & 512 s & 767735 & 135 s & 134\\
	\midrule
	\cNone & 0 & 0 s & 36 & 0 s & 0\\
	\bottomrule
\end{tabular}
\caption{Performance results for (\cref{eval:q3}).}
\label{table:experiments-res3}
\end{table}

\begin{table}
\centering
\begin{tabular}{lrrrrr} 
	\toprule
	Approach & \# traversed links & traversal time & \# triples & query evaluation time & \# results \\
	\midrule
	\ssl & 8287 & 87 s & 3881 & 1 s & 35\\
	\midrule
	\cMatch & 95 & 2 s & 1732 & 1 s & 53\\
	\midrule
	\cAll & 29968 & 454 s & 703911 & 2.4 h & 56902\\
	\midrule
	\cNone & 0 & 0 s & 286 & 0 s & 0\\
	\bottomrule
\end{tabular}
\caption{Performance results for (\cref{eval:q4}).}
\label{table:experiments-res4}
\end{table}

%% file: query-processing.tex
\section{Discussion}
\label{sec:QueryProcessing}
So far, we have studied~\ltqp from the perspective of data quality;
namely, we allow \queryingagents and/or \datapublishers
to capture a~subweb of data
that satisfies certain quality properties for them.
In real-world applications,
such quality properties could for example indicate
different notions of trust,
or something use-case-specific such as data~sensitivity levels.
While our formal framework only associates a single subweb specification to each
agent, it is not hard to extend it to associate multiple subweb constructions with each agent
and allow the \queryingagent  to pick a suitable one.
%


The same mechanism can be used to improve \emph{efficiency} in two ways:
the data publishers can opt to \emph{not} include certain documents in their subweb, and
for the ones included, they can use a \emph{filter} which indicates which data
will be used from said document.
%
%
%
%

Most prominently,
every publisher of Linked~Data typically has their own way
of organizing data across documents,
and they could capture this structure in their subweb of interest.
For~example,
in contrast to Bob (\cref{lst:Bob}),
Ann stores her profile information in multiple documents
(\cref{lst:Ann,lst:AnnDetails}).
If she were to declare this as a subweb specfication, she can use
filters to indicate which data can be found in which documents.
A~query processor can then exploit this information
to only follow links to relevant documents  (documents of Ann's subweb for which the filter \emph{could}
keep triples that contribute to the query result).
For example,
Uma's \queryingagent can use
Ann's subweb construction of \cref{lst:FollowFriendsAnnScl}
to prune the set of links to follow,
and as such perform a~guided navigation
while maintaining completeness guarantees.
Without even inspecting \url{https://photos.ex/ann/},
it~knows Ann (and thus Uma) does not trust triples in this document
for data about~her,
so fetching~it will not change the final query result.
Whereas \ltqp under \cAll~semantics
would require at~least 7~\http~requests,
the filters allow us to derive
which 4~requests are needed
to return all 3~trusted results of the specification-annotated~query.
Analogous performance gains were observed in work
on provenance-enabled queries~\cite{ProvenanceQueries}.
In contrast,
traditional~\ltqp cannot make any~assumptions
of what to encounter behind a~link.
The work
on describing document structures using shapes~\cite{spec:shapetrees}
can be~leveraged~here.
%

As such, filters in subweb specifications serve two purposes:
they define \emph{semantics} by selecting only part of a~data source,
and give query processors \emph{guidance}
for saving bandwidth and thus processing~time.

%% file: conclusion.tex
\section{Conclusion}
\label{sec:Conclusion}
\ltqp
is generally not considered suitable for real-world applications
because of its performance and data~quality implications.
However,
if the current decentralization trend continues,
we need to prepare for a~future with multi-source query processing,
since some data \emph{cannot} be~centralized
for legal or other reasons.

Federated querying over expressive interfaces such as~\sparql endpoints
only addresses part of the problem:
empirical evidence suggests that, counterintuitively,
less expressive interfaces can  lead
to faster processing times for several queries~\cite{verborgh_jws_2016},
while being less expensive to host.
A~document-based interface
is about the simplest interface imaginable,
and is thereby partly responsible for the Web's scalability.
Hence the need to~investigate
how far we can push~\ltqp
for internal and external integration of private and public~data.

Our formalization for specification-annotated queries
creates the theoretical foundations
for a~next generation of traversal-based
(and perhaps \emph{hybrid}) query processing,
in which data quality can be controlled tightly,
and network requests can be reduced significantly.
Moreover,
the efforts to realize these necessary improvements
are distributed across the~network,
because every \datapublisher can describe their own subwebs.
Importantly, the availability of such descriptions
is also driven by other needs.
For instance,
initiatives such as~Solid~\cite{verborgh_timbl_chapter_2020}
store people's personal data as Linked~Data,
requiring every personal data~space to describe
their document organization
such that applications can read and write data at the correct locations~\cite{spec:shapetrees}.

This article opens multiple avenues for future work.
A~crucial direction is the algorithmic handling of the theoretical framework, and its software implementation,
for which we have ongoing~work in the Comunica query engine~\cite{taelman_iswc_2018}; an important open question here is how the expressed filters can be exploited for query optimization.
Also on the implementation level,
the creation and management of subweb specifications should be facilitated.
This is because we don't expect users to manually create these subweb specifications,
but instead are to be created through a user-friendly user interface or automatic learning-based approaches,
similar to how today's discovery mechanisms~\cite{spec:typeindex,spec:shapetrees} in Solid are managed. 
Empirical evaluations will shed~light
on cases where subweb annotated \WOLDs and queries result in a~viable~strategy.

%% file: appendix.subwebs.tex
\label{sec:subwebs-annot}
In this section, we present the subweb specifications
attached to the virtual linked web.
Table \ref{table:exp-subwebs} contains different subweb specifications
used in the settings of the experiment, specifically for \ssl.  As for the
specifications used in \swslOne and \swslTwo, they are simple modifications
and they will be mentioned afterward.
The symbol \verb|<>| is used to refer to the agent publishing the subweb.
Person-type documents are the ones that have the most information.
A person includes information about his city, friends, university, and company.
The $\bnfts{WITH SUBWEBS}$ directive is used by Persons' subweb specifications
for friends, cities, and workplaces,
this has the effect of also including the subweb specifications of these agents
when evaluating the persons' subwebs of interest.
The $\bnfts{INCLUDE}$ directive is used to allow information
only about the agents to be included
(triples having the agent \IRI in the subject position).
A forum includes (non-recursively) the content of its
messages (posts and comments), members, and moderators.

\begin{table}
\centering
\begin{tabular}{l|l|l}
\toprule
Document (\# of Instances) &
Subweb specification &
Information retrieved \\
\midrule

\multirow{19}{*}{Person (1528)} &
\begin{lstlisting}[language=SPARQL]
FOLLOW ?city WITH SUBWEBS {
 <> voc:isLocatedIn ?city.
} INCLUDE { ?city ?p ?o. }
\end{lstlisting} &
City of the Person $^{*}$\\
\cmidrule{2-3}

&
\begin{lstlisting}[language=SPARQL]
FOLLOW ?person WITH SUBWEBS {
 <> voc:knows ?e.
 ?e voc:hasPerson ?person.
} INCLUDE { ?person ?p ?o. }
\end{lstlisting} &
Friends of the Person $^{*}$ \\
\cmidrule{2-3}

&
\begin{lstlisting}[language=SPARQL]
FOLLOW ?univ WITH SUBWEBS {
 <> voc:studyAt ?u.
 ?u voc:hasOrganisation ?univ.
} INCLUDE { ?univ ?p ?o. }
\end{lstlisting} &
University of the Person \\
\cmidrule{2-3}

&
\begin{lstlisting}[language=SPARQL]
FOLLOW ?org WITH SUBWEBS {
 <> voc:workAt ?w.
 ?w voc:hasOrganisation ?org.
} INCLUDE { ?org ?p ?o. }
\end{lstlisting} &
Company where the Person works\\
\cmidrule{2-3}

&
\begin{lstlisting}[language=SPARQL]
FOLLOW ?comment {
 <> voc:likes ?thing.
 ?thing voc:hasComment ?comment. }
\end{lstlisting} &
Comments of the Person \\
\cmidrule{2-3}

&
\begin{lstlisting}[language=SPARQL]
FOLLOW ?post {
 <> voc:likes ?thing.
 ?thing voc:hasPost ?post. }
\end{lstlisting} &
Posts of the Person \\
\midrule

\multirow{6}{*}{Forum (13750)} &
\begin{lstlisting}[language=SPARQL]
FOLLOW ?message {
 <> voc:containerOf ?message. }
\end{lstlisting} &
Posts and Comments of the Forum \\
\cmidrule{2-3}

&
\begin{lstlisting}[language=SPARQL]
FOLLOW ?member {
 <> voc:hasMember ?m.
 ?m voc:hasPerson ?member. }
\end{lstlisting} &
Members of the Forum \\
\cmidrule{2-3}

&
\begin{lstlisting}[language=SPARQL]
FOLLOW ?mod {
 <> voc:hasModerator ?mod. }
\end{lstlisting} &
Moderator of the Forum \\
\midrule

\multirow{3}{*}{Comment (151043)} &
\begin{lstlisting}[language=SPARQL]
FOLLOW ?post {
 <> voc:replyOf ?post. }
\end{lstlisting} &
Post (or Comment) of the Comment \\
\cmidrule{2-3}

&
\begin{lstlisting}[language=SPARQL]
FOLLOW ?creator {
 <> voc:hasCreator ?creator. }
\end{lstlisting} &
Creator of the Comment \\
\midrule

\multirow{1}{*}{Post (135701)} &
\begin{lstlisting}[language=SPARQL]
FOLLOW ?creator {
 <> voc:hasCreator ?creator. }
\end{lstlisting} &
Creator of the Post \\
\midrule

\multirow{1}{*}{City (1341)} &
\begin{lstlisting}[language=SPARQL]
FOLLOW ?place RECURSE 1 {
 <> voc:isPartOf ?place. }
\end{lstlisting} &
Location hierarchy of the City \\
\midrule

\multirow{1}{*}{Country (111)} &
\begin{lstlisting}[language=SPARQL]
FOLLOW ?continent {
 <> voc:isPartOf ?continent. }
\end{lstlisting} &
Continent of the Country\\
\bottomrule

\end{tabular}
\caption{The subweb specifications used in the experiment. The specifications at $^{*}$ are modified in \swslOne and \swslTwo.}
\label{table:exp-subwebs}
\end{table}

In \swslOne, we use different specifications for the
city and friends of person as follows:
\begin{itemize}
\item for the city of the person, the subweb specification is
\begin{lstlisting}[language=SPARQL]
FOLLOW ?city WITH SUBWEBS { <> voc:isLocatedIn ?city.}
INCLUDE { ?s ?p ?o. } WHERE { FILTER ( ?s=?city || ?p=voc:isPartOf ). }
\end{lstlisting}
\item for the friends of the person, the subweb specification is
\begin{lstlisting}[language=SPARQL]
FOLLOW ?person WITH SUBWEBS { <> voc:knows ?e. ?e voc:hasPerson ?person. }
INCLUDE { ?s ?p ?o. } WHERE { FILTER ( ?s=?person || ?p=voc:isPartOf ). }
\end{lstlisting}
\end{itemize}
In the schema of the virtual linked web (Figure \ref{fig:exp-schema}),
the predicate \verb|voc:isPartOf| is used to indicate the information
``a city \verb|voc:isPartOf| a country'' and ``a country  \verb|voc:isPartOf| a continent'',
this allows a city document to refer \emph{recursively}
to the second information using the $\bnfts{RECURSE}$ directive,
which allows the inclusion of triples of the form
\verb|(?city voc:isPartOf ?country)| and
\verb|(?country voc:isPartOf ?continent)| in the filtered
subweb of \verb|?city|, and hence, these triples are
going to be included in the filtered subwebs of persons.

As for the \swslTwo, we change only the specification
for the friends of person to be as follows:
\begin{lstlisting}[language=SPARQL]
FOLLOW ?person WITH SUBWEBS { <> voc:knows ?e. ?e voc:hasPerson ?person. }
INCLUDE { ?s ?p ?o. } WHERE { FILTER ( ?s=?person || ?p=voc:hasOrganisation || ?p=foaf:name ). }
\end{lstlisting}
This will allow for knowing the (names of) universities and companies of
persons, overcoming the use of an intermediate literal.